\documentclass[
  envcountsame
]{llncs}

\usepackage[utf8]{inputenc}
\usepackage[T1]{fontenc}

\usepackage[
  scaled = .96
]{PTSansNarrow}
\newcommand*{\narrowfont}{\fontfamily{PTSansNarrow-TLF}\selectfont}

\usepackage{
  amsmath,
  amssymb,
  cite,
  ebproof,
  float,
  graphicx,
  ifthen,
  kantlipsum,
  listings,
  mathpartir,
  microtype,
  multirow,
  nicefrac,
  textcomp,
  times,
  utfsym,
  xcolor,
  xparse,
  xspace
}

\usepackage[
  inline
]{enumitem}

\usepackage[
  scale = .8
]{noto-mono}

\usepackage[
  pdftex,
  colorlinks = true,
  linkcolor = blue,
  citecolor = blue,
  urlcolor = blue,
  bookmarks,
  breaklinks = true
]{hyperref}

\usepackage{cleveref}
\crefname{equation}{Eq.}{equations}
\Crefname{equation}{Equation}{Equations}
\crefname{proposition}{Prop.}{propositions}
\Crefname{proposition}{Proposition}{Propositions}
\crefname{lemma}{Lemma}{lemmata}
\Crefname{lemma}{Lemma}{Lemmata}
\crefname{listing}{Listing}{listings}
\Crefname{listing}{Listing}{Listings}
\crefname{definition}{Def.}{definitions}
\Crefname{definition}{Definition}{Definitions}
\crefname{theorem}{Thm.}{theorems}
\Crefname{theorem}{Theorem}{Theorems}
\crefname{figure}{Fig.}{figures}
\Crefname{figure}{Figure}{Figures}
\crefname{page}{p.}{pages}
\Crefname{page}{Page}{Pages}
\crefname{section}{Sect.}{sections}
\Crefname{section}{Section}{Sections}

\usepackage[normalem]{ulem}

\usepackage[draft,silent]{fixme}
\fxsetup{theme=color,mode=multiuser}
\definecolor{fxtarget}{rgb}{0.8000,0.0000,0.0000}
\FXRegisterAuthor{todo}{TODO}{TODO}
\FXRegisterAuthor{ebj}{EBJ}{Einar}
\FXRegisterAuthor{is}{IS}{Ina}
\FXRegisterAuthor{aw}{AW}{Andrzej}
\FXRegisterAuthor{rp}{RP}{Ra\'ul}

\font\domino=domino
\newcommand{\sdie}{{%
  \raisebox{-0.06 \baselineskip}
    {\resizebox{1.5ex}{!}{\domino3}}}%
  \xspace}
\newcommand{\scoin}{{\kern.5pt\usymH{2784}{1.3ex}\kern.7pt}}
\newcommand{\scoinone}{{\kern.5pt\usymH{2780}{1.3ex}\kern.7pt}}

\newcommand{\Since}[2][0pt]{\hspace{-#1}\quad\text{(\small #2)}}

\renewcommand{\implies}{\rightarrow}
\renewcommand{\phi}{\varphi}
\newcommand{\embed}[1]{\ensuremath{{[\![#1]\!]}}}
\renewcommand{\epsilon}{\varepsilon}
\newcommand{\overbar}[1]{\mkern 1.5mu\overline{\mkern-1.5mu#1\mkern-1.5mu}\mkern 1.5mu}
\newcommand{\dlbox}[1]{[#1]}

\newcommand{\final}[1]{\textrm{final}(#1)}

\newcommand{\MDP}{\mathcal{M}}
\newcommand{\dom}[1]{\ensuremath{\textnormal{dom}\,#1}}
\newcommand{\PDL}{\textnormal{\narrowfont pDL}\xspace}
\newcommand{\pGCL}{\textnormal{\narrowfont pGCL}\xspace}

\newcommand{\LIT}{\ensuremath{\textnormal{\narrowfont ATF}}}
\newcommand{\pchoice}[1]{{\,}_{#1}\!\oplus }
\newcommand{\ndchoice}{\ensuremath{\sqcap}}
\newcommand{\pbox}[2][s]{\ensuremath{[#1]_{#2}}}

\newcommand{\subst}[2]{\ensuremath{[ #1 := #2 ]}}

\NewDocumentCommand{\Expectation}%
  {O{\epsilon} D(){absent} m}%
  {\ifthenelse{ \equal {#2} {} }
    {\ensuremath{\mathbf{E}_{#1}{\left({#3}\right)}}}
    {\ensuremath{\mathbf{E}_{#1}{{#3}}}}}

\NewDocumentCommand{\ExpectedV}%
  {O{\epsilon,\pi} D(){absent} m}%
  {\ifthenelse{ \equal {#2} {} }
    {\ensuremath{\mathbb{E}_{#1}{\left({#3}\right)}}}
    {\ensuremath{\mathbb{E}_{#1}{{#3}}}}}

\newcommand{\paths}[1]{\textrm{paths}(#1)}
\newcommand{\func}[1]{{\boldsymbol{#1}}}
\newcommand{\pval}{\ensuremath{\mathit{p}}}
\newcommand{\pfun}{\ensuremath{\func{p}}}
\newcommand{\pfunepsilon}{\ensuremath{\pfun(\epsilon)}}

\newcommand{\ncondrule}[3]{
  \begin{array}{c}
    \textsc{ ({#1})} \\[1pt]
    #2 \\[1pt]
    \hline\\[-7pt]
    #3
  \end{array} }

\newcommand{\State}{\textit{State}\xspace}
\newcommand{\Act}{\textit{Act}}
\newcommand{\code}[1]{\text{\lstinline[mathescape=true]|#1|}}

\newlist{strtproof}{enumerate}{10}
\setlist[strtproof]{topsep=2ex,itemsep=2ex,wide,labelwidth=!,labelindent=0pt,label*=\arabic*.}

\title{A Specification Logic for Programs in the\\ Probabilistic Guarded Command Language \\ {\large (Extended Version)}}

\author{Ra\'ul Pardo\inst 1 \and Einar Broch Johnsen\inst 2 \and Ina Schaefer\inst 3 \and Andrzej Wąsowski\inst 1}

\institute{IT University of Copenhagen, Copenhagen, Denmark,
\and
University of Oslo, Oslo, Norway
\and
Karlsruhe Institute of Technology, Karlsruhe, Germany}

\begin{document}

\maketitle

\begin{abstract}
  The semantics of probabilistic languages has been extensively studied, but specification languages for their properties have received little attention. This paper introduces the probabilistic dynamic logic \PDL, a specification logic for programs in the probabilistic guarded command language (\pGCL) of McIver and Morgan.  The proposed logic \PDL can express both first-order state properties and probabilistic reachability properties, addressing both the non-deterministic and probabilistic choice operators of \pGCL. In order to precisely explain the meaning of specifications, we formally define the satisfaction relation for \PDL.  Since \PDL embeds \pGCL programs in its box-modality operator, \PDL satisfiability builds on a formal MDP semantics for \pGCL programs.  The satisfaction relation is modeled after PCTL, but extended from propositional to first-order setting of dynamic logic, and also embedding program fragments.  We study basic properties of \PDL, such as weakening and distribution, that can support reasoning systems.  Finally, we demonstrate the use of \PDL to reason about program behavior.
\end{abstract}

\section{Introduction}%
\label{sec:intro}

This paper introduces a specification language for probabilistic
programs.
Probabilistic programming
 techniques and systems are
becoming
increasingly important not only for machine-learning applications but
also for, e.g., random algorithms, symmetry breaking in distributed
algorithms and in the modelling of fault tolerance.
The semantics of probabilistic languages has been extensively studied,
from Kozen's seminal work \cite{kozen79} to recent research
\cite{hark20popl,kaminski19phd,stein21lics,smolka17popl}, but
specification languages for their properties have received little
attention (but see, e.g., \cite{batz22esop}).

The specification language we define in this paper is the
probabilistic dynamic logic \PDL, a specification logic for programs
in the probabilistic guarded command language \pGCL of McIver and
Morgan~\cite{mciver05book}. This programming language combines the
guarded command language of Dijkstra \cite{dijkstra76discipline}, in
which the non-deterministic scheduling of threads is guarded by
Boolean assertions, with state-dependent probabilistic choice.
Whereas guarded commands can be seen as a core language for concurrent
execution, \pGCL can be seen as a core language for probabilistic and
non-deterministic execution.

The proposed logic \PDL can express both first-order state properties
and reachability properties, addressing the non-deterministic as well
as the probabilistic choice operators of \pGCL.
Technically, \PDL is a probabilistic extension of (first-order)
dynamic logic \cite{harel00dynlog}, a modal logic in which programs
can occur within the modalities of logical formulae. The semantics of
dynamic logic is defined as a Kripke-structure over the set of
valuations of program variables.  Dynamic logic allows reachability
properties to be expressed for given (non-probabilistic) programs by
means of modalities. The probabilistic extension \PDL allows
probabilistic reachability properties to be similarly expressed.

In order to precisely explain the meaning of specifications expressed
in \PDL, we formally define the semantics of this logic in terms of a
satisfaction relation for \PDL formulae (a model-theoretic semantics).
The satisfaction relation is modeled after
PCTL~\cite{DBLP:journals/fac/HanssonJ94}, but extended from a
propositional to a first-order setting of dynamic logic, embedding
program fragments in the modalities.
Since \PDL embeds \pGCL programs in its formulae, the formalization of
\PDL satisfiability builds on a formal semantics for \pGCL programs,
which is defined by Markov Decision Processes (MDP)~\cite{puterman}.
The formalization of \PDL satisfiability allows us to study basic
properties of specifications, such as weakening and distribution.
Finally, we demonstrate how \PDL can be used to specify and reason
about program behavior. The main contributions of this paper are:%
\begin{itemize}
\item The specification logic \PDL to syntactically express
  probabilistic properties of stochastic non-deterministic programs
  written in \pGCL;
\item A model-theoretic semantics for \PDL over a simple MDP semantics
  for \pGCL programs; the satisfaction relation is modeled after PCTL,
  but extended from a propositional to a first-order setting of
  dynamic logics with embedded \pGCL programs; and
\item A study of basic properties of \PDL and a demonstration of how \PDL can
  be used to specify and reason about \pGCL programs.

\end{itemize}

\noindent
Our motivation for this work is ultimately to define a proof system
which allows us to mechanically verify high-level properties for
programs written in probabilistic programming languages. Dynamic logic
has proven to be a particularly successful logic for such verification
systems in the case of regular (non-probabilistic) programs; in
particular, KeY~\cite{key}, which is based on forward reasoning over
DL formulae, has been used for breakthrough results such as the
verification of the TimSort algorithm \cite{gouw15cav}.  The
specification language introduced in this paper constitutes a step in
this direction, especially by embedding probabilistic programs into
the modalities of the specification language.  Further, the semantic
properties of \PDL form a semantic basis for proof rules, to be
formalized, proven correct, and implemented in future work.


\section{State of The Art}%
\label{sec:related}

Verification of probabilistic algorithms has been addressed with abstract interpretation~\cite{CousotAbsIntProb}, symbolic execution~\cite{FilieriPV13}, or probabilistic model checking~\cite{KwiatkowskaNP12}. Here, we focus on logical reasoning about probabilistic algorithms using dynamic logic. Existing dynamic logics for probabilistic programs are Kozen's PPDL and PrDL of Feldman and Harel. Kozen introduces probability by drawing variable values from distributions, while propositions are measurable real-valued functions \cite{Kozen85}. The program semantics is purely probabilistic; PPDL does not include demonic choice. Probabilistic Dynamic Logic (PrDL) relies on the same notion of state, but introduces probabilistic transitions using a random choice operator\,\cite{FeldmanH82}. Since neither PPDL nor PrDL include non-determinism, to reason about non-deterministic stochastic programs in a program logic we need a new specification language. We aim to develop a first-order dynamic logic for programs (PPDL was propositional) with demonic and probabilistic choice.
\looseness = -1

The main alternative for logical reasoning about probabilistic programs is the weakest pre-expectation calculus, proposed by McIver and Morgan for the probabilistic guarded command language (\pGCL)~\cite{mciver05book}. The language contains explicit probabilistic and demonic choice. Program states are modeled by classical (non-probabilistic) variable assignments, and probabilities are introduced by an explicit probabilistic choice. Assertions are real-valued functions over program state capturing expectations, where a Boolean embedding is used to derive expectations from logical assertions. Reasoning in \pGCL follows a backwards expectation transformer semantics. McIver and Morgan define an axiomatic semantics given by the weakest pre-expectation calculus over \pGCL programs, but do not introduce an operational semantics for the language. Also they do not provide a specification language for \pGCL assertions, i.e., real-valued functions, beyond the Boolean embedding (cf.~\cite{BatzSpec2021}).  In this work, we want to build on this tradition.  However, we think there is a need for a specification language with classical model-theoretical semantics known from logics---a satisfaction semantics.  Dynamic logics is a good basis for such a development, since it is strictly more expressive than Hoare logic and weakest precondition calculi---both can be embedded in dynamic logic \cite{Haehnle22a}. In contrast to these calculi, dynamic logics are closed under logical operators such as first-order connectives and quantifiers; for example, program equivalence, relative to state formulae $\phi$ and $\psi$, can be expressed by the formula $\phi \Rightarrow \pbox[s_1]{} \psi \iff \phi \Rightarrow \pbox[s_2]{} \psi$.

As mentioned, the original \pGCL lacked operational semantics. Since semantics is needed for a traditional definition of satisfaction in a modal logic, we propose to use the  MDP semantics similar to the one of  Gretz et al.\,\cite{GretzKM14}, where post-expectations are rewards in final states. An alternative could be Kaminski's computation tree semantics\,\cite{kaminski19phd}, but we find it more complex and less standard for our purpose (deviating further from traditions of simpler logics like PCTL).
\looseness = -1

Termination analysis of probabilistic programs~\cite{mciver18popl,hark20popl} considers probabilistic reachability properties.  This and other directions of related work, such as separation logic for probabilistic programs~\cite{BatzKKMN19}, expected run-time analysis for probabilistic programs~\cite{KaminskiKMO16} and relational reasoning over probabilistic programs for sensitivity analysis~\cite{RelationalSensitivity2021}, are orthogonal to the goal of defining a specification language for programs, and thus outside of scope of interest for this particular paper. Generally all these approaches rely on the backwards pre-expectation transformer semantics of McIver and Morgan~\cite{mciver05book}.

\section{Preliminaries}
\label{sec:preliminaries}

We review the basic semantic notions used in the main part of the paper.

\begin{definition}[Markov Decision Process]\label{def:mdp}
  A \emph{Markov Decision Process} (MDP) is a tuple $M \! = \! (\State,\Act, \mathbf{P})$ where
  \looseness = -1
  \begin{enumerate*}[label=(\roman*)]

    \item $\State$ is a countable set of states,

    \item $\Act$ is a countable set of actions,

    \item $\mathbf{P} \! : \State \!\times\! \Act \rightharpoonup \textrm{Dist}(\State)$ is a partial transition probability function.
      \looseness = -1

  \end{enumerate*}
\end{definition}

\noindent
Let \(\sigma\) denote the states and $a$ the actions of an MDP. A state \(\sigma\) is \emph{final} if no further transitions are possible from it, i.e.\ \((\sigma, a) \not \in \textrm{dom} (\mathbf{P})\) for any $a$. A \emph{path}, denoted \(\overline\sigma\), is a sequence of states $\sigma_1,\ldots,\sigma_n$ such that $\sigma_n$ is final and there are actions \(a_1, \ldots, a_{n-1}\) such that \(\mathbf{P}(\sigma_i,a_i)(\sigma_{i+1})\geq 0\) for $1\leq i < n$. Let \(\final{\overbar\sigma}\) denote the final state of a path \(\overbar\sigma\).
\looseness = -1

For a given state, the set of applicable actions of \(\mathbf{P}\) defines the \emph{demonic choices} between successor state distributions.  A \emph{positional policy} \(\pi\) is a function that maps states to actions, so \(\pi: \State \rightarrow \Act\).  We assume \(\pi\) to be consistent with $\mathbf{P}$, so \(\mathbf{P} (\sigma, \pi(\sigma))\) is defined.  Given a policy \(\pi\), we define a transition relation $\xrightarrow{\cdot}_\pi \subseteq \State \times [0,1] \times \State$ on states that resolves all the demonic choices in $\mathbf{P}$ and write:
\begin{equation} \label{eq:transition-relation-iff-partial-transition-function}
  \sigma\xrightarrow{\pval}_\pi \sigma' \quad \text{ iff } \quad
  \mathbf{P}(\sigma,\pi(\sigma)) (\sigma') = \pval.
\end{equation}
\looseness = -1

\noindent
For a given policy $\pi$, we let \(\smash{\xrightarrow{\pval}}^\ast_\pi \subseteq \State \times [0,1] \times \State\) denote the reflexive and transitive closure of the transition relation, and define the probability of a path \(\overline{\sigma}=\sigma_1,\ldots,\sigma_n\) by
\begin{equation}
  p = \Pr (\overline\sigma) = 1 \cdot \pval_1 \cdots \pval_n
  \quad \text{ where }  \sigma_1 \xrightarrow{\pval_1}_\pi \cdots \xrightarrow{\pval_n}_\pi \sigma_n .
\end{equation}
Thus, a path with no transitions consists of a single state \(\sigma\), and \(\Pr(\sigma)=1\).  Let \(\textrm{paths}_\pi(\sigma)\) denote the set of all paths with policy \(\pi\) from $\sigma$ to final states.

In this paper we assume that MDPs (and the programs we derive them from) arrive at final states with probability 1 under all policies.  This means that the logic \PDL that we will be defining and interpreting over these MDPs can only talk about properties of almost surely terminating programs, so in general it cannot be used to reason about termination without adaptation. This is what corresponds to the notion of partial correctness in non-probabilistic proof systems.

An MDP may have an associated \emph{reward function} \(r: \State \to [0,1]\) that assigns a real value \(r(\sigma)\) to any final state \( \sigma \in \State \).  (In this paper we assume that rewards are zero everywhere but in the final states.) We define the \emph{expectation} of the reward starting in a state \(\sigma\) as the greatest lower bound on the expected value of the reward over all policies; so the real valued function defined as \looseness = -1
\begin{equation}
  \Expectation[\sigma]()r = \inf_\pi \ExpectedV[\sigma,\pi]()r = \inf_\pi \sum_{\overline \sigma \in \textrm{paths}_\pi(\sigma)} \!\!\!\!\!\! \Pr (\overline \sigma) \, r ( \final{\overline \sigma} )  \enspace , \label{eq:expectation}
\end{equation}
where \ExpectedV[\sigma, \pi]()r stands for the \emph{expected value} of the random variable induced by the reward function under the given policy, known as the \emph{expected reward}.  Note that the expectation \Expectation[\sigma]()r always exists and it is well defined. First, for a given policy the expected value \ExpectedV[\sigma, \pi]()r is guaranteed to exist, as we only consider terminating executions and our reward functions are bounded, non-negative, and non-zero in final states only. The set of possible positional policies that we are minimizing over might be infinite, but the values we are minimizing over are bounded from below by zero, so the set of expected values has a well defined infimum.  Finally, because the MDPs considered here almost surely arrive at a final state, we do not need to condition the expectations on terminating paths to re-normalize probability distributions, which greatly simplifies the technical machinery.

To avoid confusing expectations and scalar values, we use bold font for expectations in the sequel. For instance, \(\func p\) represents an unknown expectation from the state space into $[0,1]$, and \(\func 0\) represents a constant expectation function, equal to zero everywhere.
\looseness = -1

We use characteristic functions to define rewards for the semantics of pGCL programs, consistently with McIver \& Morgan\,\cite{mciver05book}.  For a formula $\phi$ in some logic with the corresponding satisfaction relation, a characteristic function \embed\phi, also known as a Boolean embedding or an indicator function, assigns \(1\) to states satisfying \(\phi\) and \(0\) otherwise.  In this paper, models will be program states, and also states of an MDP.  In general, characteristic functions can be replaced by arbitrary real-valued functions \cite{kaminski19phd}, but this is not needed to interpret logical specifications, so we leave this to future work.

Finally, given a formula \(\phi\) that can be interpreted over a state space of an MDP, we define the truncation of a reward function \(\pfun\) as the function \((\pfun\!\downarrow\!\phi)(\sigma) = \pfun(\sigma)\cdot\embed{\phi}(\sigma)\). The truncation of \(\pfun\) to \(\phi\) maintains the original value of \(\pfun\) for states satisfying \(\phi\) and gives zero otherwise. Note that \(\pfun\!\downarrow\!\phi\) remains a valid reward function if \(\pfun\) was.

\begin{figure}[t]\centering

  \begin{equation*}
    \begin{array}{lrl}
      v &::= & \textit{true} \mid \textit{false} \mid 0 \mid 1 \mid \ldots \\[.6mm]
      e &::=& v \mid x \mid \textit{op}\ e \mid e\ \textit{op}\ e\\[.6mm]
      op& ::= & +\; \mid \; -\; \mid \; *\; \mid \; /\; \mid \; >\; \mid \; ==\; \mid \; \geq\\[.6mm]
      s &::= & s \ndchoice s \mid s \pchoice{e} s \mid s;s \mid \code{skip} \mid x:= e
        \mid \code{if}\ e\ \{ s \}\ \code{else}\ \{ s \} \mid \code{while}\ e\ \{ s \}
    \end{array}
  \end{equation*}

  \vspace{-2.1mm}

  \caption{The syntax of the probabilistic guarded command language pGCL}
  \label{fig:pgcl}

\end{figure}

\section{pGCL: A Probabilistic Guarded Command Language}
\label{sec:pgcl}

The probabilistic guarded command language \pGCL~\cite{mciver05book}, extends Dijkstra's guarded command language \cite{dijkstra76discipline} with probabilistic choice.  \Cref{fig:pgcl} gives the syntax of \pGCL. We let $x$ range over the set $X$ of program variables, $v$ over primitive values, and $e$ over expressions $\mathit{Exp}$.  Expressions $e$ are constructed over program variables $x$ and primitive values $v$ by means of unary and binary operators $\mathit{op}$ (including logical operators $\neg, \land,\lor$ and arithmetic operators $+, -, *,/$). Expressions are assumed to be well-formed.

Statements $s$ include the non-deterministic (or demonic) choice $s_1 \sqcap s_2$ between the branches $s_1$ and $s_2$. We write $s \pchoice{e} s'$ for the probabilistic choice between the branches $s$ and $s'$; if the expression $e$ evaluates to a value $\pval$ given the current values for the program variables, then $s$ and $s'$ have probability $\pval$ and $1-\pval$ of being selected, respectively.  In many cases $e$ will be a constant, but in general it can be an \emph{expression over the state variables} (i.e., $e\in \mathit{Exp}$), so its semantics will be an real-valued function.  Sequential composition, $\code{skip}$, assignment, $\code{if-then-else}$ and $\code{while}$ are standard (e.g., \cite{dijkstra76discipline}).

The semantics of \pGCL\ programs $s$ is defined as an MDP $\MDP_s$ (cf.~\cite{GretzKM14}),
and its executions are captured by the partial transition
probability function for a given policy $\pi$, which induces the relation
$\xrightarrow{p}_\pi$ for some probability $p$, (\cref{eq:transition-relation-iff-partial-transition-function}). A state $\sigma$ of $\MDP_s$ is a pair of a
\emph{valuation} and a \emph{program}, so
$\sigma=\langle \epsilon, s\rangle$ where the valuation $\epsilon$ is
a mapping from all the program variables in $s$ to concrete values (sometimes we omit the program part, if it is unambiguous in the context).
The state $\langle \epsilon, s\rangle$ represents an \emph{initial
  state} of the program $s$ given some initial valuation $\epsilon$
and the state $\langle\epsilon, \code{skip}\rangle$ represents a
\emph{final state} in which the program has terminated with the
valuation $\epsilon$.

The rules defining the partial transition probability function for a
given policy $\pi$ are shown in Fig.~\ref{fig:semantics}.  We denote
by
$\langle\epsilon,s\rangle \xrightarrow{\pval}_\pi
\langle\epsilon',s'\rangle$
the transition from $\langle\epsilon,s\rangle$ to
$\langle\epsilon',s'\rangle$ by action
$\alpha=\pi(\langle\epsilon,s\rangle)$, where $\pval$ is the resulting
probability.  Note that for demonic choice, the policy $\pi$ fixes the
action choice between the distributions $0,1$ and $1,0$; for all
other statements, there is already a single successor distribution.
The transitive closure of this relation, denoted
$\langle\epsilon_0,s_0\rangle \smash{\xrightarrow{\pval}}^\ast_\pi
\langle\epsilon_n,s_n\rangle$,
expresses that there is a sequence of zero or more such transitions
from $\langle\epsilon_0,s_0\rangle$ to $\langle\epsilon_n,s_n\rangle$
with corresponding actions $\alpha_i =\pi(\epsilon_i,s_i)$ and
probability $\pval_i$ for $0<i\leq n$, such that
$\pval=1 \cdot \pval_1\cdots \pval_n$.

Remark that the rules in \cref{fig:semantics} allow programs to get stuck, for instance if an expression $e$ evaluates to a value outside $[0,1]$ (\textsc{ProbChoice}).  Since we are interested in partial correctness, we henceforth rule out such programs and only consider programs that successfully reduce to a single \code{skip} statement under all policies with probability 1.
\looseness = -1

\begin{figure}[t]

  \begin{equation*}
    \begin{array}{c}
      \ncondrule{Assign}{\epsilon'=\epsilon[x\mapsto \epsilon(e)]}{\langle \epsilon, x:=e\rangle \xrightarrow{1}_\pi  \langle \epsilon', \code{skip}\rangle}
      \\\\\\
      \ncondrule{Composition1}{
        \langle \epsilon, s_1 \rangle\xrightarrow{p}_\pi  \langle \epsilon',s_2\rangle}{\langle \epsilon, \code{skip}; s_1 \rangle
      \xrightarrow{p}_\pi  \langle \epsilon',s_2\rangle}
      \\\\\\
      \ncondrule{Composition2}{\langle \epsilon, s_1 \rangle \xrightarrow{p}_\pi  \langle \epsilon',s_2\rangle}{\langle \epsilon, s_1; s \rangle
      \xrightarrow{p}_\pi  \langle \epsilon',s_2;s\rangle}
    \end{array}
    \hfill
    \begin{array}{c}
      \ncondrule{ProbChoice1}{\epsilon(e)=p \quad 0 \leq p \leq 1}{\langle \epsilon, s_1\pchoice{e} s_2\rangle \xrightarrow{p}_\pi  \langle \epsilon,s_1\rangle}\\\\
      \ncondrule{ProbChoice2}{\epsilon(e)=p \quad 0 \leq p \leq 1}{\langle \epsilon, s_1\pchoice{e} s_2\rangle \xrightarrow{1-p}_\pi  \langle \epsilon,s_2\rangle}\\\\
      \ncondrule{While1}{\epsilon(e)=\textit{true}}{\langle\epsilon, \code{while}\: e\: \{
        s\}\rangle \xrightarrow{1}_\pi \langle\epsilon, s;\code{while}\: e\: \{
      s\}\rangle}\\\\
      \ncondrule{While2}{\epsilon(e)=\textit{false}}{\langle\epsilon, \code{while}\: e\: \{
      s\}\rangle \xrightarrow{1}_\pi \langle \epsilon, \code{skip}\rangle}\\[-6pt]\\
    \end{array}
    \hfill
    \begin{array}{c}
      \ncondrule{DemChoice}{ i\in\{1,2\}\\ \pi\langle\epsilon, s_1 \ndchoice  s_2\rangle = s_i}{\langle \epsilon, s_1 \ndchoice s_2\rangle \xrightarrow{1}_\pi  \langle \epsilon',s_i\rangle}\\\\
      \ncondrule{If1}{\epsilon(e)=\textit{true}}{\langle\epsilon, \code{if}\: e\: \{ s_1\}\ \code{else}\ \{ s_2\} \rangle\\  \xrightarrow{1}_\pi \langle\epsilon, s_1\rangle}\\\\
      \ncondrule{If2}{\epsilon(e)=\textit{false}}{\langle\epsilon,
        \code{if}\: e\: \{ s_1\}\ \code{else}\ \{ s_2\} \rangle
      \\\xrightarrow{1}_\pi \langle \epsilon, s_2\rangle}
    \end{array}
  \end{equation*}

  \vspace{-1mm}

  \caption{\label{fig:semantics}An MDP-semantics for pGCL.}

  \vspace{-2mm}

\end{figure}

\section{Probabilistic Dynamic Logic}%
\label{sec:pdl}

\paragraph{Formulae \& Satisfiability.}

Given sets \(X\) of program variables and \(L\) of logical variables disjoint from \(X\), let \LIT\ denote the well-formed atomic formulae built using constants, program and logical variables.  For every \(l\!\in\!L\), let \dom{l}\ denote the domain of \(l\).
We extend valuations to also map logical variables \(l\in L\) to values in \dom{l} and let \(\epsilon\models_\LIT \phi\) denote standard satisfaction, expressing that  \(\phi\in\LIT\) holds in valuation \(\epsilon\).
\looseness = -1

The formulae of probabilistic dynamic logic (\PDL) are defined inductively as the smallest set generated by the following grammar:
\newcommand{\gmid}{~\mid~}
\begin{equation}\label{eq:pdl-syntax}
  \phi \quad ::= \quad \LIT
  \gmid \neg \phi
  \gmid \phi_1\land \phi_2
  \gmid \forall l \cdot \phi
  \gmid \dlbox{s}_\pfun\:\phi
\end{equation}
where \(\phi\) ranges over \PDL formulae, \(l\! \in\! L\) over logical variables, $s$ is a \pGCL program with variables in \(X\), and \pfun\ is an expectation assigning values in $[0,1]$ to initial states of the program \(s\).  The logical operators \(\to\), \(\lor\) and \(\exists\) are derived in terms of \(\neg\), \(\wedge\) and \(\forall\) as usual.
\looseness = -1

The last operator in \cref{eq:pdl-syntax} is known as the box-operator in dynamic logics, but now we give it a probabilistic interpretation along with the name ``p-box.''  Given a pGCL program $s$, we write $[s]_\pfun\:\phi$ to express that the expectation that a formula $\phi$ holds after successfully executing $s$ is at least \pfun; i.e., the function $\pfun$ represents the expectation for $\phi$ in the current state of $\MDP_s$ using $\embed\phi$ as the reward function (see \cref{sec:preliminaries}). For the reader familiar with the CTL/PCTL terminology, the p-box formulae are path formulae, and all other formulae are state formulae.

We define semantics of \emph{well-formed formulae} in \PDL, so formulae with no free logical variables---all occurrences of logical variables are captured by a quantifier.  The definition extends the standard satisfaction relation of dynamic logic \cite{harel00dynlog} to the probabilistic case:
\looseness = -1
\begin{definition}[Satisfaction of \PDL Formulae]
  \label{def:satisfaction}
  Let \(\phi\) be a well-formed \PDL formula, $\pi$ range over policies, \(l \! \in \! L\), \(\pfun : \State \to [0,1]\) be an expectation lower bound, and $\epsilon$ be a valuation defined for all variables mentioned in $\phi$.  The \emph{satisfiability} of a formula $\phi$ in a model $\epsilon$, denoted $\epsilon\models \phi$, is defined inductively as follows:
  \begin{align*}
    & \epsilon \models \phi
    & \textnormal{iff}\quad
    & \epsilon \models_\LIT \phi \quad\textnormal{for}\quad \phi\in \LIT
    \\
    & \epsilon \models \phi_1\land \phi_2
    & \textnormal{iff}\quad
    & \epsilon \models \phi_1 \quad\textnormal{and}\quad \epsilon \models \phi_2
    \\
    & \epsilon \models \neg \phi
    & \textnormal{iff}\quad
    & \textnormal{not } \epsilon \models \phi
    \\
    & \epsilon \models \forall l \cdot  \phi
    & \textnormal{iff}\quad
    & \epsilon \models \phi \subst{l}{v} \textnormal{ for each } v\in \dom{l}
    \\
    & \epsilon \models \dlbox{s}_\pfun \phi
    & \textnormal{iff}\quad
    & \pfunepsilon \leq \Expectation{\embed \phi}
    \textnormal{ where the expectation is taken in $\MDP_s$}
  \end{align*}
\end{definition}

\noindent
For $\phi\in \LIT$, \(\models_\LIT\) can be used to
check satisfaction just against the valuation of program variables
since \(\phi\) is well-formed.  In the case of universal
quantification, the substitution replaces logical variables with
constants.  The last case (p-box) is implicitly recursive, since the
characteristic function \embed\phi\ refers to the satisfaction of
\(\phi\) in the final states of \(s\).

The satisfaction of a p-box formula $\dlbox{s}_\pfun\:\phi$ captures a
lower bound on the probability of \(\phi\) holding after the program
\(s\). Consequently, \PDL supports specification and reasoning about
probabilistic reachability properties in almost surely terminating
programs.

It is convenient to omit the valuation $\epsilon$ from the
satisfaction judgement, meaning that the judgement holds for all
valuations (validity):

\begin{equation}
  \models \pbox\pfun \, \phi \quad \text{iff} \quad
  \epsilon \models \pbox\pfun\, \phi \quad \text{for all valuations } \epsilon
\end{equation}\label{eq:validity}
\section{The p-box Modality and Logical Connectives}%
\label{sec:distribution}

We begin our investigation of \PDL by exploring how the p-box operator
interacts with different expectations and the other connectives of \PDL.

In a proof system, weakening is useful to allow adjusting proven facts
to a format of a syntactic proof rule.  Since all operators of \PDL,
with the exception of p-box, behave like in first order logic, the
usual qualitative weakening properties apply for these operators at
the top-level. For instance, $\phi_1 \land \phi_2$ can be weakened to
$\phi_1$.  These properties follow directly from
\cref{def:satisfaction}.  The following proposition states the key
properties for p-box: \looseness = -1

\begin{proposition}[Weakening]
  \label{prop:weakening}
  Let \(\epsilon\) stand for a valuation, \( \func p, \func 0 \in \State \rightarrow [0, 1] \) be expectation lower bounds, $s$ a pGCL program, and \( \phi \in \PDL \). Then:

  \begin{enumerate}

    \item Universal lower bound: \( \epsilon \models [s]_\func{0}\, \phi \) \label{item:universal-lb}

    \item Quantitative weakening: $\epsilon \models [s]_{\pfun_1}\, \phi$  then $\epsilon \models [s]_{\pfun_2}\, \phi \text{ if } \pfun_2 \leq \pfun_1$ everywhere \label{item:quant-weakening}

    \item Weakening conjunctions: $\epsilon \models \pbox\pfun\, (\phi_1 \land \phi_2)$  then $\epsilon \models \pbox\pfun \, \phi_i \text{ ~for } i = 1, 2$ \label{item:conj-weakening}

    \item Qualitative weakening: $\epsilon \models \pbox\pfun\, \phi_1$ and \( \models \phi_1 \implies \phi_2 \)  then \( \epsilon \models \pbox\pfun\, \phi_2 \) \enspace .\label{item:qual-weakening}

  \end{enumerate}

\end{proposition}

\noindent
The first point states that there is a limit to the usefulness of weakening the expectation: if you cannot guarantee that the lower bound is positive, then you do not have any information at all. A zero lower-bound would hold for any property. The second property is a probabilistic variant of weakening, which follows directly from the last case of \cref{def:satisfaction}; the lower bound on an expectation can always be lowered.  The last two properties are the probabilistic counterparts of weakening in standard (non-probabilistic) dynamic logic; the third property is syntactic for conjunction, the last one is general.
\looseness = -1

When building proofs with \PDL, the other direction of reasoning seems
more useful: we would like to be able to derive a conjunction
from two independently concluded facts.  For state
formulae, this holds naturally, like in first-order logic.  For
p-box formulae, we would like to use the expectations \(\pfun_i\) of
two formulae \(\phi_i\) to draw conclusions about the expectation that
their conjunction holds.  It seems tempting to translate the
intuitions from the Boolean lattice to real numbers, and to suggest
that a minimum of the expectations for both formulae is a lower bound
for their conjunction.  To develop some intuition, let us first
consider an incorrect proposal using the following counterexample:
\looseness = -1

\begin{example}\label{ex4}
  Consider the program $\sdie$, modeling a six-sided fair die:
  \begin{equation}
    \sdie \quad ::= \quad
      \code{x:=1} \pchoice{\nicefrac 16} (
      \code{x:=2} \pchoice{\nicefrac 15} (
      \code{x:=3} \pchoice{\nicefrac 14} (
      \code{x:=4} \pchoice{\nicefrac 13} (
      \code{x:=5} \pchoice{\nicefrac 12}
      \code{x:=6}
      ))))
  \end{equation}
  Let `odd' be an atomic formula stating that a value is odd, and `prime'
  an atomic formula stating that it is prime.  Since the die is fair,
  the expectations for each of these after \sdie are:
  \looseness = -1
  \begin{equation}
    \models [\sdie]_\func{\nicefrac 12}\, \textrm{odd} (x) \quad \quad
    \models [\sdie]_\func{\nicefrac 12}\, \textrm{prime} (x) \quad \quad
  \end{equation}
  The minimum of the two expectations is a constant function which
  equals \( \func{\nicefrac 12} \) everywhere, but the expectation bound in \( \pbox\pfun (\textrm{odd} (x) \land \textrm{prime} (x)) \) can be at most
  \( \func{\nicefrac 13} \) since only two outcomes ($x\mapsto 3$ and
  $x\mapsto 5$) satisfy both predicates.  Effectively, even if
  \(\epsilon\models \pbox{\pfun_1} \, \phi_1\) and
  \(\epsilon\models \pbox{\pfun 2} \, \phi_2\) hold, we do not
  necessarily have
  \( \epsilon \models \pbox{\min(\pfun_1, \pfun_2)} \, \phi_1 \land
  \phi_2 \).  The reason is that the expectation bounds measure what is the lower bound on satisfaction of a property, but not where in the execution space this probability mass is placed.  There is not enough information to see to what extent the two properties are overlapping. \qed

\end{example}

\noindent
Similarly,
\( \pfun(\epsilon) = \pfun_1 (\epsilon) \pfun_2 (\epsilon) \) is not a
good candidate in Ex.~\ref{ex4}, since it is only guaranteed to be a
lower bound for a conjunction when \(\phi_i\) are independent
events. Unless \(\pfun_1 \! = \! \pfun_2 \! = \! \func 1\), combining proven facts
with conjunction (or disjunction) weakens the expectation: \looseness
= -1

\begin{theorem}
  \label{thm:distributive-laws}
  Let \(\epsilon\)
  be a valuation, \( \pfun, \pfun_1, \pfun_2 \in \State \rightarrow [0, 1] \)
  expectation lower bounds, $s$ a pGCL program, and
  \( \phi_1, \phi_2 \in \PDL \). Then:

  \begin{enumerate}

    \item p-box conjunction:\label{item:p-box-conjunction}
      if \( \epsilon \models \pbox {\pfun_1}\, \phi_1 \) and \( \epsilon \models \pbox {\pfun_2}\, \phi_2 \), then \( \epsilon \models \pbox\pfun \, ( \phi_1 \land \phi_2 ) \) where \( \pfun = \max ( \pfun_1\! + \! \pfun_2 -1, 0 )\) everywhere.

    \item p-box disjunction:\label{item:p-box-disjunction-1}
      if \( \epsilon \models \pbox{\pfun_1}\, \phi_1 \) or \( \epsilon \models \pbox{\pfun_2}\, \phi_2 \), then \( \epsilon \models \pbox \pfun \, ( \phi_1 \lor \phi_2 ) \) where \(\pfun = \min (\pfun_1, \pfun_2)\) everywhere.

  \end{enumerate}

\end{theorem}

\noindent
Note the asymmetry between these cases: reasoning about conjunctions
of low probability properties using
\cref{thm:distributive-laws}.\ref{item:p-box-conjunction} is
inefficient, and quickly arrives at the lower bound expectation
\(\func{0}\), which, as observed in \cref{prop:weakening}, holds
vacuously. If both properties have an expected probability lower than
\(\nicefrac12\), then \PDL cannot really see (in a compositional manner)
whether there is any chance that they can be satisfied
simultaneously. In contrast, compositional reasoning about
disjunctions makes sense both for low and high probability events.
This is a consequence of using lower bounds on expectations. 
The bounds in~\cref{thm:distributive-laws} are consistent with prior work by Baier et al. on LTL verification of probabilistic systems~\cite{BaierKN99}.

The qualitative non-probabilistic specialization of
\cref{thm:distributive-laws}.\ref{item:p-box-conjunction} behaves reasonably: when  \(\phi_1\) or \(\phi_2\) hold almost surely, then the
theorem reduces to a familiar format: \looseness = -1
\begin{equation}
  \text{if } \epsilon \models \pbox{\pfun} \phi_1
  \text{ and } \epsilon \models \pbox{\func1} \phi_2
  \text{ then } \epsilon \models \pbox{\pfun} (\phi_1 \land \phi_2)
\end{equation}
\noindent

\begin{theorem}
  \label{thm:quantifiers-commute}
  Let \(\epsilon\) be a valuation, \( \func p \in \State \!\to\! [0, 1] \) an expectation lower bound, \(s\) a pGCL program, and \( \phi \in \PDL \) a well-formed formula.
  \begin{enumerate}

    \item If \(\epsilon \models \pbox\pfun \, \forall l \cdot\phi\)  ~then~  \(\epsilon \models \forall l \cdot \pbox\pfun \, \phi\), but not the other way around in general.
      \label{item:universal-commutes}

    \item If \(\epsilon \models \exists l \cdot \pbox\pfun \, \phi  \text{ ~then~ } \epsilon \models \pbox\pfun \, \exists l \cdot\phi\) but not the other way around in general.
      \label{item:existential-commutes}

  \end{enumerate}
\end{theorem}

\noindent
The essence of the above two properties lies in the fact that quantifiers in \PDL only affect logical variables, programs cannot access logical variables, and we do not allow quantification over expectation variables.

In a deductive proof system, one works with abstract states, not just concrete states.  A state abstraction can be introduced as a precondition, a \PDL property that captures the essence of an abstraction, and is satisfied by all the abstracted states sharing the property. If an abstract property is a precondition for a proof, it is naturally introduced using implication.  However, implication is unwieldy in an expectation calculus, so it is practical to be able to eliminate it in the proof machinery. The following theorem explains how a precondition can be folded into an expectation function:
\begin{theorem}[Implication Elimination]
  \label{thm:implication-elimination}
  Let \(s\) be a \pGCL program, \(\phi_i\) be \PDL formulae, and \pfun\ a lower-bound function for expectations. Then:
  \[
    \models \phi_1 \implies (\pbox\pfun \, \phi_2)
    \quad \text{iff} \quad
    \models \pbox{\pfun \downarrow \phi_1} \phi_2
  \]
\end{theorem}
Note that we use validity naturally when working with abstract states, as the state is replaced by the precondition in the formula.

Finally, negation in \PDL is difficult to push over boxes.  This is due to non-determinism and the lower bound semantics of expectations it enforces.  A p-box property expresses a lower bound on probability of a post-condition holding after a program.  Naturally, a negation of a p-box property will express an \emph{upper-bound} on a property, but \PDL has no upper-bound modality first-class. We return to this problem in \cref{sec:purely-probabilistic}, where we discuss reasoning about upper-bounds in non-deterministic and in purely probabilistic programs.
\looseness = -1
\section{Expectations for Program Constructs}%
\label{sec:program-props}

This section investigates how expectations are transformed by \pGCL program constructs, as opposed to logical constructs discussed above. We begin by looking at the composite statements, which build the structure of the underlying MDP. The probabilistic choice introduces a small expectation update, consistent with an expectation of a Bernoulli variable (item 1). The demonic choice (item 2), requires that both sides provide the same guarantee, which is consistent with worst-case reasoning.
\begin{theorem}[Expectation and Choices]
  \label{thm:choices}
  Let \(s_i\) be programs, \(\phi\) a PDL formula, \(\func p_i\) lower bound functions for expectations into $[0,1]$, and \(\epsilon\) a valuation of variables. Then:

  \begin{enumerate}

    \item If \( \epsilon \models \pbox[s_1]{\pfun_1}  \phi \) and \( \epsilon \models \pbox[s_2]{\pfun_2}  \phi \) then \( \epsilon \models \pbox[s_1 \pchoice{e} s_2]{\pfun} \, \phi \)
    \\
    with \(\pfun = \epsilon(e)\pfun_1 \! + \! ( 1 \! - \! \epsilon(e)) \pfun_2 \) \label{item:extension}

    \smallskip

    \item \( \epsilon \models \pbox[s_1]{\pfun} \, \phi \) and \(\epsilon \models \pbox[s_2]{\pfun} \, \phi \) if and only if \( \epsilon \models \pbox[ s_1 \ndchoice s_2 ]\pfun \, \phi \) \label{item:ndchoice}

  \end{enumerate}

\end{theorem}
Note that in the second case, demonic, we can always use weakening (\cref{prop:weakening}.\ref{item:quant-weakening}) to equalize the left-hand-side expectation lower-bounds using a point-wise minimum, if the premises are established earlier for different lower bound functions.

\begin{example}
This example shows that a non-deterministic assignment is less informative than a probabilistic assignment.
It shows that  \PDL can be used to make statements that compare programs directly in the formal system---one of its distinctive features in comparison with prior works (cf.\,\cref{sec:related}).
We check satisfaction of the following \PDL\ formula for any
expectation lower bound $\pfun$:
\looseness = -1
$$
\models \forall \delta \cdot \forall p \cdot
0\leq p \leq 1 \implies ([\text{\lstinline|x:=0|} \ndchoice \text{\lstinline|x:=1|}]_\pfun (x
\geq \delta) \implies [\text{\lstinline|x:=0|} \pchoice{p}
\text{\lstinline|x:=1|}]_\pfun (x \geq \delta))\ .
$$
For simplicity, we use the logical variable $p$ directly in the rightmost program (this can easily be encoded as an additional assumption equating a fresh logical variable to a program variable).
For the proof, we first simplify the formula using equivalence rewrites:
\begin{alignat*}{3}
  & \models \forall \delta \cdot \forall p \cdot
  0\leq p \leq 1 \implies
  \rlap{\ensuremath{([\text{\lstinline|x:=0|} \ndchoice \text{\lstinline|x:=1|}]_\pfun (x
\geq \delta) \implies [\text{\lstinline|x:=0|} \pchoice{p}
\text{\lstinline|x:=1|}]_\pfun (x \geq \delta))}}
  \\
  & \text{iff for \(\epsilon\), \(\delta\), \(p\) we have }\\
  & \epsilon \models 0 \!\leq\! p \!\leq\! 1 \!\implies\!
  ([\text{\lstinline|x:=0|} \ndchoice \text{\lstinline|x:=1|}]_\pfun (x
\geq \delta) \implies [\text{\lstinline|x:=0|} \pchoice{p}
  \text{\lstinline|x:=1|}]_\pfun (x \geq \delta))
  & \Since{\cref{eq:validity}, \cref{def:satisfaction} \(\forall\)}
  \\
  & \text{iff for \(\epsilon\), \(\delta\), \(p\) we have }\\
  & \epsilon \models \neg 0 \!\leq\! p \!\leq\! 1  \!\lor\!
  \neg[\text{\lstinline|x:=0|} \ndchoice \text{\lstinline|x:=1|}]_\pfun (x
\geq \delta) \lor [\text{\lstinline|x:=0|} \pchoice{p}
  \text{\lstinline|x:=1|}]_\pfun (x \geq \delta)
  & \Since{syntactic sugar}
  \\
  & \text{iff for \(\epsilon\), \(\delta\), \(p\) we have }\\
  & \neg 0 \!\leq\! p \!\leq\! 1  \,\lor\,
  \neg \pfun(\epsilon) \leq \Expectation[\epsilon](){x
  \geq \delta} \,\lor\, \pfun(\epsilon) \leq \Expectation[\epsilon](){x \geq \delta}
  & \Since{\cref{def:satisfaction}, the box}
\end{alignat*}
In the last line above the left expectation is taken in MDP \(\MDP_{\text{\lstinline|x:=0|} \ndchoice \text{\lstinline|x:=1|}}\) and the right one is taken in \(\MDP_{\text{\lstinline|x:=0|} \pchoice{p} \text{\lstinline|x:=1|}}\).

Now the property is a disjunction of three cases. If the first or second disjunct hold the formula holds vacuously (the assumptions in the statement are violated). We focus on the last case, when the first two disjuncts are violated (so the assumptions hold). We need to show that the last disjunct holds. We split the reasoning in two cases:
\begin{enumerate}
  \item \(\delta \leq 0\): Consider the right expectation
    \(\Expectation[\epsilon](){x \geq \delta}\). In the right program
    this expectation is equal to \(1\) because the formula always
    holds (both possible values of \(x\) are greater or equal to \(\delta\)).  Consequently, any expectation lower bound \(\pfun\) is correct for this formula: \(\pfun(\epsilon) \leq \Expectation[\epsilon](){x \geq \delta}\) in the right program.

  \item \(\delta > 0\): Consider the left expectation
    \(\Expectation[\epsilon](){x \geq \delta}\). By
    \cref{eq:expectation} this expectation is equal to zero (the
    policy that chooses the left branch in the program violates the
    property as $x = 0 < \delta$). Since
    \(\pfun(\epsilon) \leq \Expectation[\epsilon](){x \geq \delta} =
    0\), it must be that \(\pfun(\epsilon) = 0\) in the left
    program. By the universal lower bound property
    (\cref{prop:weakening}.1), all properties hold after any program
    with the expectation lower bound \(\pfun\), including the
    post-condition of the right program.  \looseness = -1 \qed
\end{enumerate}
\end{example}

\noindent
For any program logic, it is essential that we can reason about composition of consecutive statements; allowing the post-condition of one to be used as a pre-condition for the other.  The following theorem demonstrates that sequencing in \pGCL corresponds to composition of expectations in the MDP domain. It uses implication elimination (\cref{thm:implication-elimination}) to compute a post-condition for a sequence of programs. Crucially, the new lower bound is computed using an expectation operation in the MDP of the first program, using the lower-bound of the second program as a reward function.  Here, the expectation operation acts as a way to explore the program graph and accumulate values in final states.
\looseness = -1

\begin{theorem}[Expectation and Sequencing]
  \label{thm:sequencing}
  Let \(s_i\) be \pGCL programs, \(\phi_i\) be \PDL formulae, \(\epsilon\) be a valuation, and \(\pfun\) an expectation lower bound function.
  \begin{center}
    If \(\models \phi_1 \implies (\pbox[s_2]{\pfun} \, \phi_2) \)  then  \(\epsilon \models \pbox[s_1;s_2]{\Expectation[\langle\epsilon,s_1\rangle](){\pfun\downarrow\phi_1}} \, \phi_2 \) \enspace,
  \end{center}
  where the expectation \(\Expectation[\langle\epsilon,s_1\rangle](){\pfun\!\downarrow\!\phi_1}\) is taken in \(\MDP_{s_1}\) with \(\pfun\!\downarrow\!\phi_1\) as the reward function.

    \label{item:decompose}

\end{theorem}
For a piece of intuition, note that the above theorem captures the basic step of a backwards reachability algorithm for MDPs, but expressed in \PDL; it accumulates expectations backwards over \(s_1\) from what is already known for \(s_2\).

We now move to investigating how simple statements translate expectations:
\begin{theorem}[Unfolding Simple Statements]
  \label{thm:simple-statements}
  Let \(s\) be a \pGCL program, \(\phi\) a \PDL formula, \(\pfun\) a function into \([0,1]\), a lower bound on expectations, and \(\epsilon\) a valuation. Then
  \looseness = -1
  \begin{enumerate}

    \item \(\epsilon \models \pbox[\code{skip}]{\func{1}} \, \phi\)  iff \(\epsilon \models \phi\)\label{item:skip}

    \smallskip

    \item \(\epsilon \models \pbox\pfun \, \phi\) ~ iff ~ \(\epsilon \models \pbox[ \code{skip}; s ]\pfun \, \phi \)
    \label{item:skip-composed}

    \smallskip

    \item  \( \epsilon \models \pbox[ \code{x:=} e; s ]\pfun \, \phi \) ~ iff ~ \( \epsilon [x \mapsto \epsilon(e)] \models \pbox\pfun \, \phi \)
    \label{item:assignment}

  \end{enumerate}
\end{theorem}
The case of if-conditions below is rather classic (\cref{thm:loops-conditionals}.\ref{item:while-unfold}).  For any given state, we can evaluate the head condition and inherit the expectation from the selected branch.  For this to work we assume that the atomic formulae (\LIT) satisfaction semantics in \PDL is consistent with the expression evaluation semantics in \pGCL.  The case of while loops is much more interesting---indeed a plethora of works have emerged recently on proposing sound reasoning rules for while loop invariants, post-conditions and termination (see \cref{sec:related}).  In this paper, we show the simplest possible reasoning rule for loops that performs a single unrolling, exactly along the operational semantics. Of course, we are confident that many other rules for reasoning about while loops (involving invariants, prefixes, or converging chains of probabilities) can also be proven sound in \PDL---left as future work.
\looseness = -1
\newcommand{\whileunfolded}{\code{if } \;e\; \code{\{} s  \code{; while } \;e\;\; \code{\{} s \, \code{\}\} else \{skip\}}}
\newcommand{\whilefolded}{\, \code{while } \;e\; \code{\{} s \, \code{\}}}
\newcommand{\IfThenElse}{\code{if} \; e \;\;\code{\{} s_1 \, \code{\} else \{ } s_2 \, \code{ \}}}
\begin{theorem}[Unfolding Loops and Conditionals]
  \label{thm:loops-conditionals}
  Let \(e\) be a program expression (also an atomic \PDL formula over program variables in \(X\)), \(\phi\) be a \PDL atomic formula, \(s_i\) be \pGCL programs, \(\pfun\) an expectation lower bound function, and \(\epsilon\) a valuation. Then:
  \begin{enumerate}

    \item If \( \epsilon \models e \land  \pbox[ s_1 ]\pfun \, \phi \)  then \( \epsilon \models \pbox[ \IfThenElse ]\pfun \, \phi \)\label{item:if-true}

      \smallskip

    \item If \( \epsilon \models \neg e \land  \pbox[ s_2 ]\pfun \, \phi \)  then \(\epsilon \models \pbox[\code{if } \;e\;\;  \code{\{} s_1 \, \code{\}  else \{} s_2 \, \code{\}}]\pfun \, \phi \)\label{item:if-false}

      \smallskip

    \item \(\epsilon \models \pbox[ \whileunfolded ]\pfun \, \phi\) iff \(\epsilon \models \pbox[ \whilefolded ]\pfun \, \phi\)
    \label{item:while-unfold}

  \end{enumerate}

\end{theorem}
\section{Purely Probabilistic and Deterministic Programs}%
\label{sec:purely-probabilistic}

The main reason for the lower-bound expectation semantics in \PDL (inherited from McIver\&Morgan) is the presence of demonic choice in \pGCL.  With non-determinism in the language, calculating precise probabilities is not possible.  However, this does not mean that \PDL cannot be used to reason about upper-bounds.  The following theorem explains:\footnote{The theorem is named as a tribute to the song \emph{Both sides now} by Joni Mitchell.}
\looseness = -1
\begin{theorem}[Joni's Theorem]
  For a policy \(\pi\), property \(\phi\), program \(s\), and state \(\epsilon\): if \( \epsilon \models \pbox{\pfun_1} \phi \) and  \( \epsilon \models \pbox{\pfun_2} \neg\phi\) then \( \mathbb E_{ \pi, \epsilon } \embed\phi \in [ p_1, 1-p_2 ] \).

\end{theorem}
The theorem means that for a purely probabilistic program derived by fixing a policy for a \pGCL program $s$, the expected reward is bounded from below by the expectation of this reward in \(s\), and from above by the expectation of its negation in \(s\). The theorem follows directly from \cref{eq:expectation} and the negation case in \cref{def:satisfaction}.

For deterministic programs, some surprising properties, follow from interaction of probability and logics.  For instance, we can conclude a conjunction of \emph{expectations} from an expectation of a disjunction.

\begin{theorem}\label{thm:disjunction-to-conjunction}

  Let $s$ be a purely probabilistic pGCL program (a program that does not use the demonic choice), let \(\epsilon\) stand for a valuation, \( \func p \in \State \rightarrow [0, 1] \) be an expectation function, and \( \phi_i \in \PDL \) properties.   Then if \(\epsilon \models \pbox{\pfun} \, ( \phi_1 \lor \phi_2 ) \) then there exist \(\pfun_1\), \(\pfun_2\), \( \pfun_1 + \pfun_2 \geq \pfun \) everywhere, such that \( \epsilon \models \pbox{\pfun_1}\, \phi_1 \) and \( \epsilon \models \pbox{\pfun_2}\, \phi_2 \).

\end{theorem}

\noindent
Intuitively, the property holds, because each of the measure of the space of final states of the disjointed properties can be separated between the disjuncts.  This separation would not be possible with non-determinism, as shown in the following counterexample.

\begin{example}

  Consider the program \(\scoin ::= \code{x := H} \,\ndchoice\, \code{x := T} \). The following holds for any initial valuation \(\epsilon\):
  \[
    \epsilon \models \pbox[\scoin]{\func1} (x=\texttt{H} \lor x = \texttt{T})
  \]
  This happens because disjunction is weakening and a weaker property is harder to avoid, here impossible to avoid, for an adversary minimizing an expectation satisfaction. However, at the same time: \(\epsilon \models \pbox[\scoin]{\func{0}} (x = \text{\ttfamily H}) \text{ and } \epsilon \models \pbox[\scoin]{\func{0}} (x = \text{\ttfamily T}) \text{ and } \func0 + \func0 < \func1\).  Importantly, zero is the tightest expectation lower bound possible here.
  \qed
\end{example}
\section{Program Analysis with pDL}
\label{sec:case-studies}

In this section, we apply \PDL\ to reason about two illustrative examples: the Monty Hall game~(\cref{subsec:monty-hall}), and convergence of a Bernoulli random variable~(\cref{subsec:bernoulli}).

\subsection{Monty Hall Game}
\label{subsec:monty-hall}

In this section, we use \PDL\ to compute the probability of winning the \emph{Monty Hall game}.
In this game, a host presents 3 doors, one of which contains a prize and the others are empty, and a contestant must figure out the door behind which the prize is hidden.
To this end, the host and contestant follow a peculiar sequence of steps.
First, the location of the prize is non-deterministically selected by the host.
Secondly, the contestant chooses a door.
Then, the host opens an empty door from those that the contestant did not choose.
Finally, the contestant is asked whether she would like to switch doors.
We determine, using \PDL, what option increases the chances of winning the prize (switching or not).

\Cref{code:monty-hall} shows a \pGCL\ program, \lstinline|Monty_Hall|, modeling the behavior of host and contestant.
There are 4 variables in this program: \lstinline|prize| (door containing the prize), \lstinline|choice| (door selected by the contestant), \lstinline|open| (door opened by the host), \lstinline|switch| (Boolean indicating whether the user switches in the last step).
Note that the variable \lstinline|switch| is undefined in the program.
The value of \lstinline|switch| encodes the strategy of the contestant, so its value will be part of a \PDL\ specification that we study below.
Line 1 models the hosts's non-deterministic choice of the door for the prize.
Line 2 models the door choice of the contestant (uniformly over the 3 doors).
Lines 3-6 model the selection of the door to open, from the non-selected doors by the contestant.
Lines 7-10 model whether the contestant switches door or not.
For clarity and to reduce the size of the program, in lines 6 and 8, we use a shortcut to compute the door to open and to switch, respectively.
Note that for $x,y \in \{0,1,2\}$ the expression $z = (2x-y)\;\mathbf{mod}\;3$ simply returns $z \in \{0,1,2\}$ such that $z \not= x$ and $z\not=y$.
Similarly, in line 4, the expressions $y=(x+1)\;\mathbf{mod}\;3, z=(x+2)\;\mathbf{mod}\;3$ ensure that $y\not=x$, $z\not=x$ and $y \not= z$.
This shortcut computes the doors that the host may open when the contestant's choice (line 2) is the door with the prize.

\begin{lstlisting}[mathescape=true,xleftmargin=.25\textwidth,label={code:monty-hall},caption={Monty Hall Program (\lstinline|Monty_Hall|)}]
prize := 0 $\ndchoice$ (prize := 1 $\ndchoice$ prize := 2);
choice := 0 $\pchoice{1/3\;}$ (choice:=1 $\pchoice{1/2\;}$ choice:=2);
if (prize = choice)
    open := (prize+1)%3 $\ndchoice$ open := (prize+2)%3;
else
    open := (2*prize-choice)%3;
if (switch)
    choice := (2*choice-open)%3
else
    skip
\end{lstlisting}

\noindent
We use \PDL\ to find out the probability of the contestant selecting the door with the prize.
To this end, we check satisfaction of the following formula, and solve it for \pfun.
\begin{equation}
\epsilon[\textit{switch} \mapsto \textit{true}] \models [\text{\lstinline|Monty_Hall|}]_\pfun(\textit{choice} = \textit{prize}).
\end{equation}
\noindent
First, we show that $\pfun = \min(\pfun_0,\pfun_1,\pfun_2)$ where each $\pfun_i$ is the probability for the different locations of the prize.
Formally, we use \cref{thm:choices}.\ref{item:ndchoice} (twice) as follows
\begin{align*}
&\epsilon  \models [\text{\lstinline|prize:=0;...|}]_{\pfun_0}(\textit{choice} = \textit{prize}) \textit{ and } \\
&\epsilon  \models [\text{\lstinline|prize:=1;...|}]_{\pfun_1}(\textit{choice} = \textit{prize}) \textit{ and } \\
&\epsilon  \models [\text{\lstinline|prize:=2;...|}]_{\pfun_2}(\textit{choice} = \textit{prize})   \textit{ imply } \\
&\epsilon  \models [\text{\lstinline|Monty_Hall|}]_{\min(\pfun_0,\pfun_1,\pfun_2)}(\textit{choice} = \textit{prize})
\end{align*}

\noindent
For each $\pfun_i$, we compute the probability for each branch of the probabilistic choice.
To this end, we use \cref{thm:choices}.\ref{item:extension} as follows:
\begin{align*}
&\epsilon \models [\text{\lstinline|choice:=0;...|}]_{\pfun_{i0}}(\textit{choice} = \textit{prize}) \; \textit{ and } \\
&\epsilon \models [\text{\lstinline|(choice:=1|} \pchoice{1/2} \text{\lstinline|choice:=2);...|}]_{\pfun_{i1}}(\textit{choice} = \textit{prize}) \; \textit{ imply } \\
&\epsilon \models [\text{\lstinline|choice:=0|} \pchoice{1/3} \text{\lstinline|(choice:=1|} \pchoice{1/2} \text{\lstinline|choice:=2|}\text{\lstinline|);...|}]_{1/3\cdot\pfun_{i0} + 2/3\cdot\pfun_{i1}}(\textit{choice} = \textit{prize}).
\end{align*}
and apply it again for $\pfun_{i1}$ to resolve the inner probabilistic choice:
\begin{align*}
&\epsilon \models [\text{\lstinline|choice:=1;...|}]_{\pfun_{i10}}(\textit{choice} = \textit{prize}) \; \textit{ and } \\
&\epsilon \models [\text{\lstinline|choice:=2;...|}]_{\pfun_{i11}}(\textit{choice} = \textit{prize}), \; \textit{implies } \\
&\epsilon \models [\text{\lstinline|(choice:=1|} \pchoice{1/2} \text{\lstinline|choice:=2);...|}]_{1/2\cdot\pfun_{i10} + 1/2\cdot\pfun_{i11}}(\textit{choice} = \textit{prize})
\end{align*}
These steps show that $\pfun_i = 1/3\cdot\pfun_{i0} + 2/3\cdot1/2\cdot\pfun_{i10} + 2/3\cdot1/2\cdot\pfun_{i11}$ where $\pfun_{i0}$, $\pfun_{i10}$ and $\pfun_{i11}$ are the probabilities for the paths with $\textit{choice}$ equals to 0, 1 and 2, respectively.

Let us focus on the case $\pfun_1$.
This is the case when the prize is behind door 1, $\epsilon[\textit{prize} \mapsto 1]$.
In what follows, we explore the three possible branches of the probabilistic choice.
Consider the case where the user chooses door 1, i.e., $\epsilon[\textit{choice} \mapsto 1]$ and
$$
\epsilon \models [\text{\lstinline[mathescape=true]|if (prize = choice) \{$s_0$\} else \{$s_1$\};...|}]_{\pfun_{110}}(\textit{choice}=\textit{prize})
$$
where \lstinline[mathescape=true]|$s_0$| and \lstinline[mathescape=true]|$s_1$| correspond to lines 4 and 6 in~\cref{code:monty-hall}, respectively.
Since $\epsilon \models \textit{prize} = \textit{choice}$ holds and by \cref{thm:loops-conditionals}.\ref{item:if-true} we derive that
$$
\epsilon \models [\text{\lstinline[mathescape=true]|$s_0$;...|}]_{\pfun_{110}}(\textit{choice}=\textit{prize}).
$$
Note that $\pfun_{110}$ remains unchanged.
Statement \lstinline[mathescape=true]|$s_1$| contains a non-deterministic choice, so we apply \cref{thm:choices}.\ref{item:ndchoice} to derive $\pfun_{110} = \min(\pfun_{1100},\pfun_{1101})$ where each $\pfun_{110i}$ correspond to the cases where $\epsilon[\textit{open} \mapsto 2]$ and $\epsilon[\textit{open} \mapsto 0]$, respectively.
Since $\textit{switch} = \textit{true}$ both branches execute line 8, and the probabilities remain the same (\cref{thm:loops-conditionals}.\ref{item:if-true}).
A simple calculation shows that after executing line 8 $\epsilon \not\models (\textit{prize} = \textit{choice})$ for both cases.
For instance, consider
$$
\epsilon[\textit{open} \mapsto 0] \models [\code{choice := (2*choice-open)\%3}]_{\pfun_{1100}}(\textit{prize} = \textit{choice}).
$$
By~\cref{thm:simple-statements}.\ref{item:assignment}
$\epsilon[\textit{choice} \mapsto (2*1-0)\%3 = 2]$, which results in $\textit{prize} \not= \textit{choice}$.
By the universal lower bound rule (\cref{prop:weakening}.\ref{item:universal-lb}) we derive $\pfun_{1100} = 0$.
The same derivations show that $\pfun_{1101} = 0$, and, consequently, $\pfun_{110} = 0$.

The same reasoning shows that $\textit{prize}=\textit{choice}$ holds for the cases where $\textit{choice} \not=1$ in line 2, i.e., $\pfun_{i0}$ and $\pfun_{i11}$---we omit the details as they are analogous to the steps above.
In these cases, by \cref{thm:simple-statements}.\ref{item:skip} we derive that $\pfun_{i0} = 1$ and $\pfun_{i11} = 1$.
Recall that $\pfun_{110} = 0$ (see above), then we derive that $\pfun_1 = 1/3\cdot 1 \; + \; 2/3\cdot1/2 \cdot 0 \; + \; 2/3\cdot1/2\cdot 1$.
Consequently, $\pfun_1 = 1/3 + 1/3 = 2/3$.
Analogous reasoning shows that all $\pfun_i = 2/3$.

To summarize, the probability of choosing the door with the prize when switching is at most 2/3.
%
%
In other words, we have proven that switching door maximizes the probability of winning the prize.

\subsection{Convergence of a Bernoulli random variable}
\label{subsec:bernoulli}

We use \PDL\ to study the convergence of a program that estimates the expectation of a Bernoulli random variable.
To this end, we compute the probability that an estimated expectation is above an error threshold $\delta > 0$.
This type of analysis may be of practical value for verifying the implementation of estimators for statistical models.

Consider the following \pGCL\ program for estimating the expected value of a Bernoulli random variable (Technically the program computes the number of successes out of \(n\) trials, and we will put the estimation into the post-condition):
\begin{lstlisting}[mathescape=true,xleftmargin=.25\textwidth,label={code:bernoulli},caption={Bernoulli Program (\lstinline|Bernoulli|)}]
i := 0; c := 0;
while (i < $n$) {
  s := 0 $\pchoice{\mu\,}$ s := 1;
  c := c + s;
  i := i + 1
}
\end{lstlisting}

\noindent
Intuitively, \lstinline|Bernoulli| computes the average of $n$ Bernoulli trials $X_i$ with mean $\mu$, i.e., $X = \sum_i X_i/n$.
It is well-known that $\mathrm{E}[X] = \mu$ (e.g.,~\cite{intro-probability}).
Each $X_i$ can be seen as a sample or measurement to estimate $\mu$.
A common way to study convergence is to check the probability that the estimated mean $X$ is within some distance $\delta > 0$ of $\mu$, i.e., $\Pr(|X-\mu| > \delta)$.
In \lstinline|Bernoulli|, a sample $X_i$ corresponds to the execution of the  probabilistic choice $\pchoice{\mu\,}$ in line 3 of \cref{code:bernoulli}.
After running all loop iterations, variable $c$ contains the sum of all the samples, i.e., $c = \sum_i X_i$.
Thus, $X$ is equivalent to $c/n$ and the specification of convergence can be written as $\Pr(|c/n - \mu| > \delta)$.
Note that this specification is independent of the implementation of the program.
The same specification can be used for any program estimating $\mu$---by simply replacing $X$ with the term estimating $\mu$ in the program.

In \PDL, we can study the convergence of this estimator by checking
$$
\epsilon \models [\text{\lstinline|Bernoulli|}]_\pfun(|c/n - \mu| > \delta)
$$
for some value of $\mu \in [0,1], \, \delta > 0$ and $n \in \mathbb{N}$.
Note that, since the program contains no non-determinism, $\pfun = \Pr(|X-\mu| > \delta)$. We describe the reasoning to compute $\pfun$.

First, note that the while-loop in \lstinline|Bernoulli| is bounded.
Therefore, we can replace it with a sequence of $n$ iterations of the loop body.
Let $s_i$ denote the $i$th iteration of the loop (lines 3-4 in \cref{code:bernoulli}).
We omit for brevity the assignments in line 1 of \cref{code:bernoulli} and directly proceed with a state $\epsilon[c \mapsto 0, i \mapsto 0]$.
Consider the first iteration of the loop, i.e., $i=0$.
By~\cref{thm:loops-conditionals}.\ref{item:while-unfold} we can derive
$$
\epsilon \models [\text{\lstinline[mathescape=true]|if (0 < n) \{$s_0$; while (i < n) \{$s_1$\}\} else \{skip\}\}|}]_\pfun(|c/n - \mu| > \delta).
$$
Assume $\epsilon \models 0 < n$ holds, then by~\cref{thm:loops-conditionals}.\ref{item:if-true} we derive
$$
\epsilon \models [\text{\lstinline[mathescape=true]|$s_0$; while (i < n) \{$s_1$\}|}]_\pfun(|c/n - \mu| > \delta).
$$
By applying the above rules repeatedly we can rewrite \lstinline|Bernoulli| as
\begin{align*}
\epsilon
\models [\text{\lstinline[mathescape=true]|$s_0$;...;$s_{n-1}$;skip|}]_\pfun(|c/n - \mu| > \delta)
\label{eq:bernoulli-init}
\end{align*}
with the \lstinline|skip| added in the last iteration of the loop by \cref{thm:loops-conditionals}.\ref{item:while-unfold} and \ref{thm:loops-conditionals}.\ref{item:if-false}.

Second, we compute the value of $\pfun$ for a possible path of \lstinline|Bernoulli|.
Consider the case when $c=0$ after executing the program.
That is,
$$\epsilon \models [\text{\lstinline|$s_0$;...;$s_{n-1}$;skip|}]_\pfun(c=0).$$
This only happens for the path where the probabilistic choice is resolved as \lstinline|c:=0| for all loop iterations.
Applying~\cref{thm:sequencing} we derive
\begin{align*}
\textit{If } \models (c=0) \to [\text{\lstinline[mathescape=true]|$s_1$;$\ldots$;$s_{n-1}$;skip|}]_{\pfun'}&(c = 0), \textit{ then } \\
&\epsilon \models [\text{\lstinline|s|}_0;\ldots;\text{\lstinline|s|}_{n-1};\text{\lstinline|skip|}]_{\Expectation{\pfun'\downarrow(c=0)}}(c = 0).
\end{align*}
Here $\Expectation{\,}$ is computed over $\MDP_{\text{\lstinline[mathescape=true]|$s_0$|}}$ (cf.~\cref{thm:sequencing}).
For \lstinline|Bernoulli|, this expectation is computed over the two paths resulting from the probabilistic choice in~\cref{code:bernoulli}, line 3.
Since only the left branch satisfies $c=0$ and it is executed with probability $\mu$, then $\Expectation{\pfun'} = \mu \pfun'$.
Applying this argument for each iteration of the loop we derive that $\epsilon \models [\text{\lstinline[mathescape=true]|$s_0$;...;$s_{n-1}$;skip;|}]_\pfun(c=0)$ holds for $\pfun = \mu^n$.
Similarly, consider the case where $c=1$ after running all iterations of the loop, due to the first iteration resulting in \lstinline|c:=1| and the rest \lstinline|c:=0|.
Then, we apply \cref{thm:sequencing} as follows
\begin{align*}
\textit{If } \models (c=1) \to [\text{\lstinline[mathescape=true]|$s_1$;$\ldots$;$s_{n-1}$;skip|}]_{\pfun'}&(c = 1), \textit{ then } \\
&\epsilon \models [\text{\lstinline|s|}_0;\ldots;\text{\lstinline|s|}_{n-1};\text{\lstinline|skip|}]_{\Expectation{\pfun'\downarrow(c=1)}}(c = 1).
\end{align*}
In this case, $\Expectation{\pfun'} = (1-\mu) \pfun'$, as the probability of $c=1$ is $(1-\mu)$ (cf.~\cref{code:bernoulli}~line 3).
Since, in this case, the remaining iterations of the loop result in \lstinline|c:=0|, and from our reasoning above, we derive that $\pfun' = \mu^{n-1}$.
Hence, $\pfun = (1-\mu)\mu^{n-1}$.
In general, by repeatedly applying these properties, we can derive that the probability of a path is $\mu^i(1-\mu)^j$ where $i$ is the number loop iterations resulting in \lstinline|c:=0| and $j$ the number of loop iterations resulting in \lstinline|c:=1|.

\begin{figure}[t!]
  \centering
  \includegraphics[width=.6\textwidth]{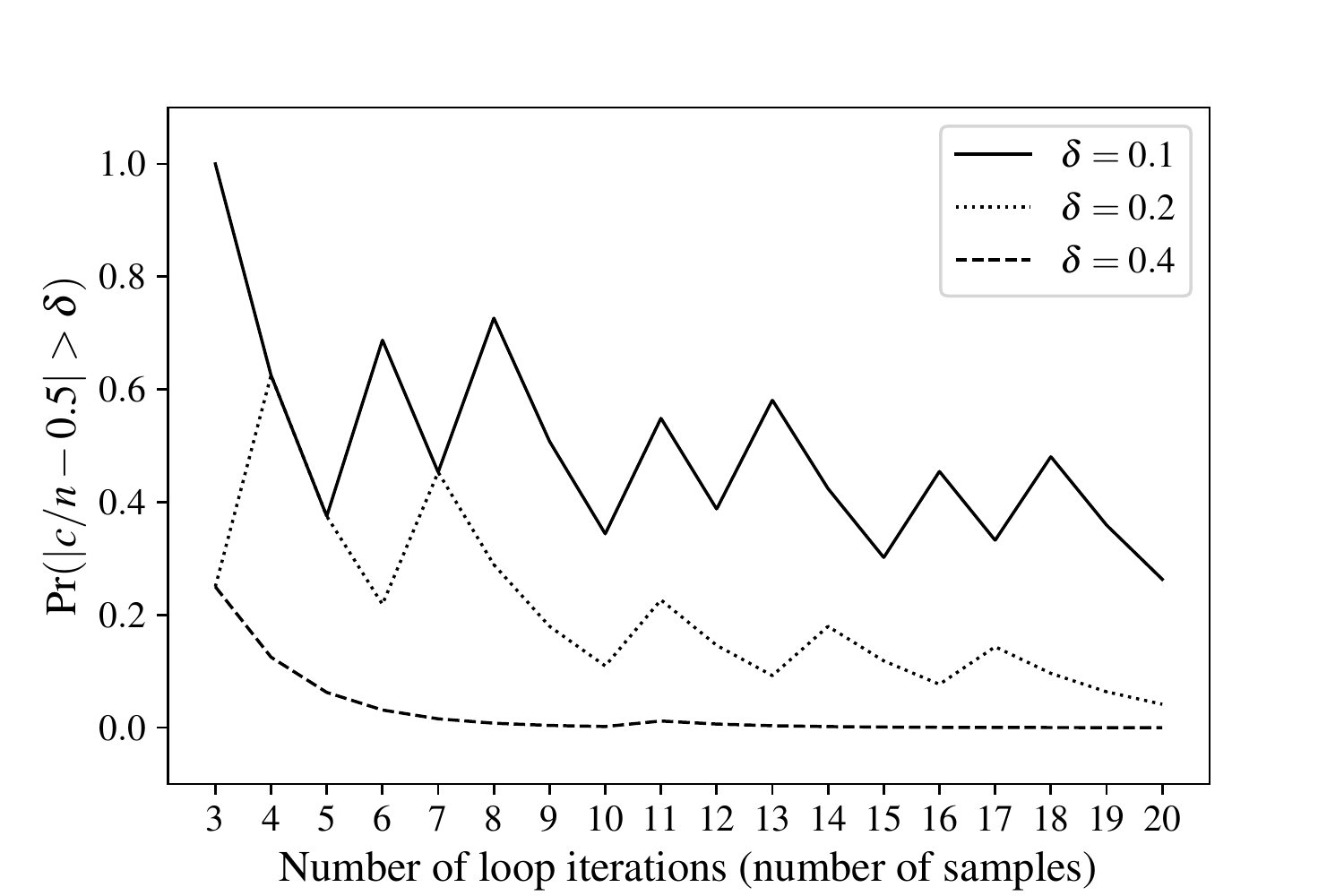}
  \caption{\label{fig:convergence} Convergence of Bernoulli random variable with $\mu=0.5$.}
\end{figure}

Now we return to our original problem $\epsilon \models [\code{Bernoulli}]_\pfun(|c/n - \mu| > \delta)$.
Recall from~\cref{def:satisfaction} that $\pfun$ is the sum of the probabilities over all the paths that satisfy the post-condition.
\lstinline|Bernoulli| has $2^n$ paths (two branches per loop iteration).
Therefore, we conclude that
$
\pfun = \sum_{i \in \Phi} \mu^{\textit{zeros(i)}}(1-\mu)^{\textit{ones(i)}}
$
where $\textit{zeros}(\cdot), \textit{ones}(\cdot)$ are functions returning the number of zeros and ones in the binary representation of the parameter, respectively, and $\Phi = \{\,i \in 2^n \, \mid \, |\textit{ones}(i)/n - \mu| > \delta\,\}$ enumerates all paths in the program satisfying the post-condition.
Note that the binary representation of $0,\ldots,2^n$ conveniently captures each of the possible executions of \lstinline|Bernoulli|.

The result above is useful to examine the convergence of \lstinline|Bernoulli|.
It allows us to evaluate the probability of convergence for increasing number of samples and different values of $\mu$ and $\delta$.
As an example, \cref{fig:convergence} shows the results for $\mu=0.5$, $\delta \in \{0.1,0.2,0.4\}$ and up to $n=20$ iterations of the loop.
The dotted and dashed lines in the figure show that with $20$ iterations the probability of having an error $\delta>0.2$ is less than $5\%$.
However, for an error $\delta > 0.1$ the probability increases to more than $20\%$.
\section{Conclusion}

This paper has proposed \PDL, a specification language for probabilistic programs---the first dynamic logic for probabilistic programs written in \pGCL. Like \pGCL, \PDL contains probabilistic and demonic choice. Unlike \pGCL, it includes programs as first-order entities in specifications and allows forward reasoning capabilities as usual in dynamic logic. We have defined the model-theoretic semantics of \PDL and shown basic properties of the newly introduced p-box modality. We demonstrated the reasoning capabilities on two well-known examples of probabilistic programs. In the future, we plan to develop a deductive proof system for \PDL supported by tools for (semi-)automated reasoning about \pGCL programs.  Furthermore, the current definition of \PDL gives no syntax to the expectations. Batz et al.\ propose a specification language for real-valued functions that is closed under the construction of weakest pre-expectations \cite{BatzSpec2021}; such a language could be used to express assertions for \pGCL programs. It would be interesting to integrate these advances into \PDL.

\paragraph{Acknowledgments} This work was supported by the Research Council of Norway via SIRIUS (project no. 237898).

\bibliographystyle{splncs}
\bibliography{refs}

\appendix

\vfill
\pagebreak

\section{Proofs}


\subsection{Proofs and Auxiliary Properties for \cref{sec:preliminaries}}

\paragraph{Expectations in an MDP.} We start with a few basic order and distribution properties of expectations that are useful in later proofs.

\begin{lemma}\label{lemma:expectation-properties}

  Given a program \(s\), the MDP \(\mathcal{M}_s\) representing its semantics, its state \( \sigma \), and reward functions \(\func{r}\), \(\func{r}_1\), \(\func{r}_2\) (into \([0,1]\) we have that

  \begin{enumerate}

    \item \( \Expectation[\sigma](){\func r} \) is a non-negative function for any non-negative reward function \(\func r\).
    \label{item:expectation-nonnegative}

    \item If \(\func r_1 \leq \func r_2\) everywhere  then \(\Expectation[\sigma](){\func r_1} \leq \Expectation[\sigma](){\func r_2}\) everywhere.
    \label{item:expectation-domination}

    \item \( \Expectation[\sigma](){\func r_1} + \Expectation[\sigma](){\func r_2} \leq \Expectation[\sigma](){\func r_1 + \func r_2} \) everywhere.
      \label{item:expectation-monotonic-plus}

    \item For a constant \(c \in \mathbb R\) we have that \( \max [\Expectation[\sigma](){\func r} + c, 0] \leq \Expectation[\sigma]{[\max (\func r + c, 0)]} \)
    \label{item:expectation-add-const}

    \item For a constant \(c \in [0,1]\) we have that \( c  \Expectation[\sigma](){\func r} = \Expectation[\sigma]{[ c  \func{r}]} \)
    \label{item:expectation-mult-const}

    \item \((1 - \inf_\pi \func{r}) = \sup_\pi \, (1-\func{r})\)%
    \text{ ~and~ } \((1 - \sup_\pi \func{r}) = \inf_\pi \, (1-\func{r})\)
    \label{item:inf-sup}

  \item \( \ExpectedV[\sigma,\pi](){\lambda \sigma' \cdot \ExpectedV[\pi,\sigma'](){\func r}} = \ExpectedV[\sigma,\pi](){\func r}\)
    \label{item:nested-expected-vs}

  \item \( \Expectation[\sigma](){\lambda \sigma' \cdot \Expectation[\sigma'](){\func r}} \leq \Expectation[\sigma](){\func r}\)
    \label{item:nested-expectations}

  \end{enumerate}

\end{lemma}

\begin{proof}[\Cref{lemma:expectation-properties}]

  \begin{enumerate}

    \item[\ref{lemma:expectation-properties}.\ref{item:expectation-monotonic-plus}] We have equality in this property if there is no non-determinism (a Markov Chain instead of an MDP). Then the property is just a regular property of expected value of random variables.  For an MDP we have:
    \begin{alignat*}{3}
      \Expectation[\sigma](){\func r_1} + \Expectation[\sigma](){\func r_2}
      & = \left[ \inf_\pi \ExpectedV[\sigma,\pi](){\func r_1}  \right] +
      \left[ \inf_\pi \ExpectedV[\sigma,\pi](){\func r_2}  \right]
      & \Since{def.\ of expectation, \cref{eq:expectation}}
      \\
      & \leq \inf_\pi \left[ \ExpectedV[ \sigma, \pi ](){\func r_1} +
      \ExpectedV[ \sigma, \pi ](){\func r_2}  \right]
      & \Since[5mm]{\(\inf f + \inf g \leq \inf (f+g)\)}
      \\
      & = \inf_\pi \ExpectedV[ \sigma, \pi ](){\func r_1 + \func r_2}
      & \Since{\ExpectedV[]\ distributes with +}
      \\
      & = \Expectation[\sigma](){\func r_1 + \func r_2}
      & \Since{\cref{eq:expectation}}
    \end{alignat*}

    \medskip

    \item[\ref{lemma:expectation-properties}.\ref{item:expectation-add-const}] We have equality if \(c\)  is non-negative.
    \begin{alignat*}{3}
      \max \left( \Expectation[\sigma](){\func r} + c, 0 \right)
      & = \max \left [ \left ( \inf_\pi \ExpectedV[\sigma, \pi](){\func r}  \right ) + c, 0 \right ]
      & \Since{\cref{eq:expectation}}
      \\
      & = \max \left[\inf_\pi \left(\ExpectedV[\sigma, \pi](){\func r} + c \right), 0 \right]
      & \Since{shift $\inf$ by a const.}
      \\
      & = \inf_\pi \max \left [ \ExpectedV[\sigma, \pi](){\func r} + c, 0 \right ]
      & \Since{argue by cases of max}
      \\
      & = \inf_\pi \max [ \ExpectedV[\sigma, \pi](){\func r + c} , 0 ]
      & \Since[2mm]{shift random variable by $c$}
      \\
      & \leq \inf_\pi \max \left \{ \ExpectedV[\sigma, \pi]{[\max (\func r + c,0)]} , 0 \right \}
      & \Since{truncate random var}
      \\
      & = \inf_\pi  \ExpectedV[\sigma, \pi]{[\max (\func r + c,0)]}
      & \Since[20mm]{\ExpectedV[\sigma, \pi]{[\max (\func r +c ,0)]} nonneg.}
      \\
      & = \Expectation[\sigma]{[\max ( \func r + c, 0 )]}
      & \Since{\cref{eq:expectation}}
    \end{alignat*}

    \medskip

    \item[\ref{lemma:expectation-properties}.\ref{item:expectation-mult-const}] The equality follows from the fact that infimum and expected value both commute with a multiplication by a non-negative constant.

    \medskip

    \item[\ref{lemma:expectation-properties}.\ref{item:nested-expected-vs}] First a comment what the lambda notation in parentheses means: we mean that the random variable in the outer expected value is itself an expected value in final states of the outer Markov Chain (which are the initial states of the inner Markov Chain). As usual a random variable in a Markov Chain (or a reward in a Markov Chain) is a function that takes as an argument the state it is calculated from. We explicitly name this argument \(\sigma\). Alternatively, we could have put a dot (\(\cdot\)) instead of \(\sigma\) in the subscript of the inner expected value: \(\lambda \sigma \cdot \ExpectedV[\pi,\sigma](){\func r}\) means the same as \(\ExpectedV[\pi,\cdot](){\func r}\).

    Now that the notation is out of the way, let's prove the lemma by a sequence of equalities:
    \begin{alignat*}{3}
      & \ExpectedV[\sigma,\pi](){\lambda \sigma' \cdot \ExpectedV[\sigma', \pi](){\func r}}
      & \Since{assumption LHS}
      \\
      & =  \sum_{\overline\sigma_1\in \paths{\sigma} } \Pr (\overbar\sigma_1)  \ExpectedV[\final{\overline\sigma_1}, \pi](){\func r}
      & \Since{\cref{eq:expectation}}
      \\
      & =  \sum_{\overline\sigma_1\in \paths{\sigma} } \hspace{-3mm} \Pr (\overbar\sigma_1)  \left [ \sum_{\overline\sigma_2\in\paths{\final{\overline\sigma_1}}} \hspace{-6mm} \Pr(\overline\sigma_2)  \func{r}(\final{\overline\sigma_2}) \right ]
      & \Since{\cref{eq:expectation}}
      \\
      & =  \sum_{\overline\sigma_1\in \paths{\sigma} } \sum_{\overline\sigma_2\in\paths{\final{\overline\sigma_1}}} \hspace{-3mm} \Pr (\overbar\sigma_1)   \Pr(\overline\sigma_2)  \func{r}(\final{\overline\sigma_2})
      & \Since{distribute multiplication}
      \\
      & =  \sum_{\overline\sigma_1\overline\sigma_2\in \paths{\sigma} }  \hspace{-3mm} \Pr (\overline\sigma_1\overline\sigma_2)   \func{r}(\final{\overline\sigma_1\overline\sigma_2})
      & \Since[12mm]{concatenate paths, now in \(\MDP_{s_1; s_2}\)}
      \\
      & = \ExpectedV[\sigma,\pi](){\func r}
      & \Since{\cref{eq:expectation}}
    \end{alignat*}

    Note that when we expand expected values using \cref{eq:expectation} the policy \(\pi\) disappears in the notation---it is implicitly included in the probability mass function \(\Pr\). We chose not to subscript it for readability. Also the paths set starting in \(\sigma\) in the penultimate row refer to the paths in \(\MDP_{s_1;s_2}\) (so the set of the longer paths).

    \medskip

    \item[\ref{lemma:expectation-properties}.\ref{item:nested-expectations}] The use of the lambda notation is the same as in the previous point, so see above.
    \begin{alignat*}{3}
      & \Expectation[\sigma](){\lambda \sigma' \cdot \Expectation[\sigma'](){\func r}}
      & \Since{assumption LHS}
      \\
      & = \inf_{\pi} \ExpectedV[\sigma,\pi](){\lambda \sigma' \cdot \inf_{\pi'} \ExpectedV[\sigma',\pi'](){\func r}}
      & \Since{\cref{eq:expectation} twice}
      \\
      & \leq \inf_{\pi} \ExpectedV[\sigma,\pi](){\lambda \sigma' \cdot \ExpectedV[\sigma',\pi](){\func r}}
      & \Since{infimum over a smaller set of policy pairs}
      \\
      & \leq \inf_{\pi} \ExpectedV[\sigma,\pi](){\func r}
      & \Since{\cref{lemma:expectation-properties}.\ref{item:nested-expected-vs} above}
      \\
      & = \Expectation[\sigma](){\func r}
      & \Since{\cref{eq:expectation}}
    \end{alignat*}

    \end{enumerate}
    \qed

\end{proof}

\paragraph{Boolean embeddings (characteristic functions).}

For these properties it is useful to equate Boolean formulae with sets of states satisfying them (since these properties make sense for any classical logic, not necessarily PDL).  The logical connectives then translate to set algebra in standard manner (conjunction is a intersection, etc.)

\begin{lemma}\label{lemma:embeddings}
  Consider formulae (sets of states), \(\phi\), \(\phi_1\) and \(\phi_2\).  Then

  \begin{enumerate}

    \item \(\embed{\phi_1} + \embed{\phi_2} - 1 \leq \embed{\phi_1 \land \phi_2} \enspace \text{everywhere}\)\label{item:weak-morgan}

    \smallskip

    \item \( \embed{\neg \phi} = 1 - \embed\phi\) ~~and~~ \( \embed{ \phi } = 1 - \embed{\neg \phi}\) ~~everywhere \label{item:negation}

  \end{enumerate}
\end{lemma}

\begin{proof}[\Cref{lemma:embeddings}]

  \begin{enumerate}

    \item Consider an element of \State which is both in \(\phi_1\) and in \(\phi_2\). Then the left hand side is 1, and so is the right hand side. The left hand side decreases by 1 or 2 for all other classes of elements, while the right hand side decreases by 1, so the inequality holds.
    \looseness = -1

    \medskip

    \item  \embed{\phi} is a zero-one function, and the construction just inverts the assignment of zeroes and ones, which is exactly what a negation does.

  \end{enumerate}
  \qed
\end{proof}

\subsection{Proofs for \cref{sec:distribution}}

\begin{proof}[\Cref{prop:weakening}]

  \begin{enumerate}

    \item[\ref{prop:weakening}.\ref{item:universal-lb}.] By \cref{lemma:expectation-properties}.\ref{item:expectation-nonnegative} and the last case of \cref{def:satisfaction}.

    \medskip

  \item[\ref{prop:weakening}.\ref{item:quant-weakening}.] By assumption \( \pfun_2 (\epsilon) \leq \pfun_1 (\epsilon) \leq \Expectation[\langle \epsilon, s\rangle]{\embed\phi} \) for any state (the latter by the last case in \cref{def:satisfaction}). Thus, by \cref{def:satisfaction}, the last case again, \( \epsilon \models \pbox{\pfun_2} \, \phi \).

    \medskip

    \item[\ref{prop:weakening}.\ref{item:conj-weakening}.] Note that \( \models \phi_1 \land \phi_2 \implies \phi_i \) and see the next point.

    \medskip

  \item[\ref{prop:weakening}.\ref{item:qual-weakening}.] Note that \( \embed{ \phi_1 } \leq \embed{ \phi_2 } \) everywhere. By the last case in \cref{def:satisfaction} and \cref{lemma:expectation-properties}.\ref{item:expectation-domination} we have \(\pfun(\epsilon) \leq \Expectation[\langle \epsilon, s \rangle]{\embed{\phi_1}} \leq \Expectation[\langle \epsilon, s \rangle]{\embed{\phi_2}} \), hence \( \epsilon \models \pbox{\pfun}\, \phi_2 \).
    \looseness = -1

  \end{enumerate}
  \qed

\end{proof}

\begin{proof}[\Cref{thm:distributive-laws}]

  \begin{enumerate}

    \item[\ref{thm:distributive-laws}.\ref{item:p-box-conjunction}.] We show that satisfaction can be concluded with the last case of \cref{def:satisfaction}:
    \begin{alignat*}{3}
      \pfun(\epsilon)
      & = \max ( \pfun_1 (\epsilon) + \pfun_2  (\epsilon) - 1, 0 ) \\
      & \leq \max ( \Expectation[\langle \epsilon, s \rangle]{\embed{\phi_1}} + \Expectation[\langle \epsilon, s \rangle]{\embed{\phi_2}} - 1, 0 )
      & \Since{\(\pfun_i\) lowerbounds}
      \\
      & \leq \max ( \Expectation[\langle \epsilon, s \rangle](){\embed{\phi_1} + \embed{\phi_2}} - 1, 0 )
      & \Since{\cref{lemma:expectation-properties}.\ref{item:expectation-monotonic-plus}}
      \\
      & \leq \Expectation[\langle \epsilon, s \rangle](){\max (\embed{\phi_1} + \embed{\phi_2} - 1, 0)}
      & \Since{\cref{lemma:expectation-properties}.\ref{item:expectation-add-const}}
      \\
      & \leq \Expectation[\langle \epsilon, s \rangle](){\max (\embed{\phi_1 \land \phi_2}, 0)}
      & \Since{\cref{lemma:embeddings}.\ref{item:weak-morgan} and \cref{lemma:expectation-properties}.\ref{item:expectation-domination}}
      \\
      & = \Expectation[\langle \epsilon, s \rangle]{\embed{\phi_1 \land \phi_2}}
      & \Since{\(\max(\embed{\phi_1 \land \phi 2}, 0) = \embed{\phi_1 \land \phi_2}\)}
    \end{alignat*}
    The \( \max \) operator is used to ensure that the obtained reward function is non-negative, which we wanted because we only work with rewards created by Boolean embeddings. This way the resulting expectation always has values in \( [0,1] \).

    \medskip

  \item[\ref{thm:distributive-laws}.\ref{item:p-box-disjunction-1}.] Assume, without loss of generality, that \(\epsilon \models \pbox{\pfun_1}\,\phi_1\) holds.  We show that satisfaction can be concluded with the last case of \cref{def:satisfaction}.
  \begin{alignat*}{3}
    \pfun(\epsilon) = \min (\pfun_1(\epsilon), \pfun_2(\epsilon))
    & \leq \pfun_1 (\epsilon)
    & \Since{minimum}
    \\
    & \leq \Expectation[\langle \epsilon, s \rangle]{\embed{\phi_1}}
    & \Since{\cref{def:satisfaction}, \(\textstyle \epsilon \models \pbox{\pfun_1}\phi_1\)}
    \\
    & \leq \Expectation[\langle \epsilon, s \rangle]{\embed{\phi_1 \lor \phi_2}}
    & \Since{\cref{prop:weakening}.\ref{item:qual-weakening}}
  \end{alignat*}

  \end{enumerate}
  \qed

\end{proof}

\begin{proof}[\Cref{thm:quantifiers-commute}]

  \begin{enumerate}

    \item[\ref{thm:quantifiers-commute}.\ref{item:universal-commutes}]  First observe that for any value \(v \in \dom l\) of a logical variable \(l \in L\) we have the following \emph{syntactic equality}:
  \[
    \left ( \pbox\pfun \, \phi \right ) \subst{l}{v} \quad \equiv \quad
    \pbox\pfun \,\left (  \phi \subst{l}{v}  \right )
  \]
  The two formulae are identical because the program \(s\) cannot refer to logical variables (\(l\)). Now we build the argument using this fact:
  \begin{alignat*}{3}
    & \epsilon \models \pbox\pfun \, \left ( \forall l \cdot  \phi \right )
    & \Since[5mm]{assumption}
    \\
    & \text{implies } \epsilon \models \pbox\pfun \left( \phi \subst{l}{v} \right)
    & \Since[2mm]{pick \(v \! \in \! \dom l\), \cref{lemma:expectation-properties}.\ref{item:expectation-domination} as \(\embed{\forall l \cdot \phi} \leq \embed{\phi[l:=v]}\)}
    \\
    & \text{iff } \epsilon \models \left( \pbox\pfun \, \phi \right) \subst{l}{v}
    & \Since{syntactic equality, above}
  \end{alignat*}
  Now observe that we have shown that \(\epsilon \models \pbox\pfun \, \left ( \forall l \cdot  \phi \right )\) implies \( \epsilon \models \left( \pbox\pfun \, \phi \right) \subst{l}{v} \) for arbitrary \(v \in \dom l\). By \cref{def:satisfaction} this means that it also implies \(\epsilon \models \forall l \cdot \pbox\pfun \, \phi\).

  \medskip

  Consider a counter example for the opposite direction. Let \scoinone be a program modeling a fair coin: \(\scoinone ::= \code{x := H} \,\pchoice{\func{\nicefrac12}}\, \code{x := T} \) and let \(\dom l = \{ \code{H}, \code{T}\}\), and the following two properties. The first property \(\phi_1\) holds in any environment, the second property \(\phi_2\) holds in no environment.
    \begin{equation}
      \phi_1 \equiv \forall l \cdot \pbox[\scoinone]{\func{\nicefrac12}} \, (x=l)
      \quad
      \quad
      \quad
      \phi_2 \equiv \pbox[\scoinone]{\func{\nicefrac12}} (\forall l \cdot x=l)
    \end{equation}

  \clearpage

  \item[\ref{thm:quantifiers-commute}.\ref{item:existential-commutes}]
    We first prove the positive case of the theorem:
    \begin{alignat*}{3}
      & \epsilon \models \exists l \cdot \pbox\pfun \, \phi
      & \Since{assumption}
      \\
      & \text{iff } \epsilon \models \neg \forall  l \cdot \neg\pbox\pfun \, \phi
      & \Since{syntactic sugar}
      \\
      & \text{iff } \text{not } \epsilon \models \forall l \cdot \neg\pbox\pfun \, \phi
      & \Since{\cref{def:satisfaction}}
      \\
      & \text{iff } \text{not for all \(v \!\in\! \dom l\):~ } \epsilon \models \neg\pbox\pfun \, (\phi \subst lv)
      & \Since[20mm]{\cref{def:satisfaction}, synt.\ equality above}
      \\
      & \text{iff } \text{not for all \(v \!\in\! \dom l\) not:~ } \epsilon \models \pbox\pfun \, (\phi \subst lv)
      & \Since{\cref{def:satisfaction}}
      \\
      & \text{iff } \text{exists \(v \!\in\! \dom l\):~ }
      \pfun(\epsilon) \leq \inf_\pi \ExpectedV[\langle \epsilon, s \rangle, \pi]{\embed{\phi \subst lv}}
      & \Since[30mm]{\cref{def:satisfaction}, meta-exists}
      \\
      & \text{iff } \text{exists \(v \!\in\! \dom l\):~ }
      \pfun(\epsilon) \leq 1 - (1 - \inf_\pi \ExpectedV[\langle \epsilon, s \rangle, \pi]{\embed{\phi \subst lv}})
      & \Since{algebra}
      \\
      & \text{iff } \text{exists \(v \!\in\! \dom l\):~ }
      \pfun(\epsilon) \leq 1 - \sup_\pi \left ( 1 - \ExpectedV[\langle \epsilon, s \rangle, \pi]{\embed{\phi \subst lv}} \right )
      & \Since{\cref{lemma:expectation-properties}.\ref{item:inf-sup}}
      \\
      & \text{iff } \text{exists \(v \!\in\! \dom l\):~ }
      \pfun(\epsilon) \leq 1 - \sup_\pi \ExpectedV[\langle \epsilon, s \rangle, \pi](){1 - \embed{\phi \subst lv}}
      & \Since[3mm]{sum expected vals}
      \\
      & \text{iff } \text{exists \(v \!\in\! \dom l\):~ }
      \pfun(\epsilon) \leq 1 - \sup_\pi \ExpectedV[\langle \epsilon, s \rangle, \pi]{\embed{\neg \phi \subst lv}}
      & \Since{\cref{lemma:embeddings}.\ref{item:negation}}
      \\
      & \text{implies } \text{exists \(v \!\in\! \dom l\):~ }
      \pfun(\epsilon) \leq 1 - \sup_\pi \ExpectedV[\langle \epsilon, s \rangle, \pi]{\embed{\forall l \cdot \neg \phi}}
      &
      \\
      &
      & \Since[30mm]{as \(\embed{\forall l \cdot \neg \phi} \leq \embed{\neg \phi\subst lv}\) for any \(v\)}
      \\
      & \text{iff } \pfun(\epsilon) \leq 1 - \sup_\pi \ExpectedV[\langle \epsilon, s \rangle, \pi]{\embed{\forall l \cdot \neg \phi}}
      & \Since[30mm]{drop free quantifier}
      \\
      & \text{iff } \pfun(\epsilon) \leq \inf_\pi \ExpectedV[\langle \epsilon, s \rangle, \pi](){1 - \embed{\forall l \cdot \neg \phi}}
      & \Since{\cref{lemma:expectation-properties}.\ref{item:inf-sup}}
      \\
      & \text{iff } \pfun(\epsilon) \leq \inf_\pi \ExpectedV[\langle \epsilon, s \rangle, \pi]{\embed{\neg \forall l \cdot \neg \phi}}
      & \Since{\cref{lemma:embeddings}.\ref{item:negation}}
      \\
      & \text{iff } \pfun(\epsilon) \leq \Expectation[\langle \epsilon, s \rangle]{\embed{\exists l \cdot \phi}}
      & \Since[30mm]{syntactic sugar, \cref{eq:expectation}}
      \\
      & \text{iff } \epsilon \models \pbox\pfun (\exists l \cdot \phi)
      & \Since{\cref{def:satisfaction}}
    \end{alignat*}

  Consider a counter example for the opposite direction. Recall the program \scoin\ modeling a ``non-deterministic coin:`` \(\scoin ::= \code{x: = H} \,\sqcup\, \code{x := T} \) and \(\dom l = \{ \code{H}, \code{T}\}\), and the following two properties. The first property \(\phi_1\) holds in any environment, the second property \(\phi_2\) holds in no environment.
    \begin{equation}
      \phi_1 \equiv \pbox[\scoin]{\func1} (\exists l \cdot x=l)
      \quad
      \quad
      \quad
      \phi_2 \equiv \exists l \cdot \pbox[\scoin]{\func1} (x=l)
    \end{equation}

  \end{enumerate}
  \qed

\end{proof}

\begin{proof}[\Cref{thm:implication-elimination}]
  We prove the co-occurrence of the two validities by splitting in cases, based on whether a particular valuation \(\epsilon\) satisfies the precondition or not.

  \medskip

  \noindent
  Case 1: \(\epsilon \models \phi_1 \) (works in both directions):
  \begin{alignat*}{3}
    & \epsilon \models \phi_1 \implies (\pbox\pfun \, \phi_2)
    & \Since{thm assumption}
    \\
    & \epsilon \models \pbox\pfun \, \phi_2
    & \Since{case}
    \\
    & \epsilon \models \pbox{\pfun\downarrow\phi_1} \, \phi_2
    & \Since{\((\pfun\downarrow\phi_1)(\epsilon) = \pfun(\epsilon)\)}
  \end{alignat*}

  \noindent
  Case 2: not \(\epsilon \models \phi_1 \) then the left-hand-side holds vacuously.   The right-hand-side also holds because the expectation is zero \((\pfun\downarrow\phi_1)(\epsilon) = \pfun(\epsilon)\cdot (\embed{\phi_1}(\epsilon)) =  \pfun(\epsilon) \cdot 0 = 0\)
  which means the formula holds by \cref{prop:weakening}.\ref{item:universal-lb}.
  \qed

\end{proof}

\subsection{Proofs for \cref{sec:program-props}}

\begin{proof}[\Cref{thm:choices}]

  \begin{enumerate}

    \item[\ref{thm:choices}.\ref{item:extension}] Let \(i\) denote the program \(s_1 \pchoice{e} s_2\). From left to right, we need to show that \(\epsilon(e)\func{p}_1(\epsilon) + \! (1 \! - \! \epsilon(e))\func{p}_2(\epsilon) \leq \Expectation[\langle \epsilon, i \rangle]{\embed\phi}\) where the expectation is taken in \(\MDP_{s_1 \pchoice{\!e\,} s_2}\).
    \begin{alignat*}{3}
      & \epsilon(e)\func{p}_1(\epsilon) + (1-\epsilon(e))\func{p}_2(\epsilon)
      &
      \\
      & \leq \epsilon(e)\Expectation[\langle \epsilon, s_1 \rangle]{\embed\phi} + (1-\epsilon(e))\Expectation[\langle \epsilon, s_2 \rangle]{\embed\phi}
      & \Since[10mm]{left in \(\MDP_{s_1}\), right in \(\MDP_{s_2}\) by assumption}
      \\
      & = \Expectation[\langle \epsilon, s_1 \rangle]{\left[ \epsilon(e)\embed\phi \right]} + \Expectation[\langle \epsilon, s_2 \rangle]{ [ (1-\epsilon(e))\embed\phi ] }
      & \Since[10mm]{\cref{lemma:expectation-properties}.\ref{item:expectation-mult-const} twice}
      \\
      & \leq \Expectation[\langle \epsilon, i \rangle]{ \left [ \epsilon(e){\embed\phi^{(1)}} \! + (1\!-\!\epsilon(e)){\embed\phi^{(2)}} \right ] }
      & \Since[10mm]{\cref{lemma:expectation-properties}.\ref{item:expectation-monotonic-plus}, \(\embed\phi^{(k)}\!\!\) positive only for \(s_k\) states}
      \\
      & = \Expectation[\langle \epsilon, i \rangle]{\embed\phi}
      & \Since[44mm]{now in \(\MDP_{s_1\pchoice{e}s_2}\) extending expected value as per \textsc{ProbChoice1/2}}
    \end{alignat*}

    \medskip

    \item[\ref{thm:choices}.\ref{item:ndchoice}] Let \(i\) denote the program \(s_1 \ndchoice s_2\). For the proof from left to right, we need to show that \(\pfun(\epsilon) \leq \Expectation{\embed{\phi}}\) where the expectation is taken in the \( \MDP_{s_1 \ndchoice s_2} \):
  \begin{alignat*}{2}
    & \pfun(\epsilon)
    \\
    & \leq \min \{ \Expectation[\langle \epsilon, s_1 \rangle]{\embed\phi}, \Expectation[\langle \epsilon, s_2 \rangle]{\embed\phi} \}
    & \Since[10mm]{left in \(\MDP_{s_1}\), right in \(\MDP_{s_2}\) by assumption}
    \\
    & = \min \{ \inf_{\pi_1}\ExpectedV[\langle \epsilon, s_1 \rangle, \pi_1]{\embed\phi}, \inf_{\pi_2}\ExpectedV[\langle \epsilon, s_2 \rangle, \pi_2]{\embed\phi} \}
    & \Since{left in \(\MDP_{s_1}\), right in \(\MDP_{s_2}\) by \cref{eq:expectation}}
    \\
    & = \inf \{ \inf_{\pi_1}\ExpectedV[\langle \epsilon, s_1 \rangle, \pi_1]{\embed\phi}, \inf_{\pi_2}\ExpectedV[\langle \epsilon, s_2 \rangle, \pi_2]{\embed\phi} \}
    & \Since{inf is minimum in a finite set}
    \\
    & = \inf_\pi \ExpectedV[\langle \epsilon, i \rangle, \pi]{\embed\phi}
    & \Since[50mm]{inf of infima, \textsc{DemChoice} has two policies, no \(\epsilon\) change}
    \\
    & = \Expectation[\langle \epsilon, i \rangle]{\embed\phi}
    & \Since{in \(\MDP_{s_1 \ndchoice s_2}\), \cref{eq:expectation}}
  \end{alignat*}
  The argument in the opposite direction follows by the same equalities backwards, and then the fact that a lower bound of a minimum is small than each element in a set.
  \looseness = -1
\end{enumerate}
  \qed
\end{proof}

\begin{proof}[\Cref{thm:sequencing}]
  Let \(i\) denote the program \(s_1 ; s_2\). To show that
  \[
    \epsilon \models \pbox[s_1; s_2]{\Expectation[\langle \epsilon, s_1 \rangle](){\pfun\downarrow \phi_1}} \, \phi_2
  \]
  we need to demonstrate that
  \[
    \Expectation[\langle \epsilon, s_1 \rangle](){\pfun\downarrow\phi_1} \leq \Expectation[\langle \epsilon, i \rangle]{\embed{\phi_2}} \enspace ,
  \]
  where the left expectation is taken in \(\MDP_{s_1}\) and the right expectation is taken in \(\MDP_{s_1; s_2}\). We use the theorem's assumption \(\models \phi_1 \implies (\pbox[s_2]\pfun\,\phi_2)\) to get \(\models  \pbox[s_2]{\pfun\downarrow\phi_1}\,\phi_2\) by \cref{thm:implication-elimination}. This in turn means that \((\pfun\downarrow\phi_1)(\epsilon')\leq \Expectation[\langle \epsilon', s_2 \rangle]{\embed{\phi_2}}\) for all \(\epsilon'\) (with the latter expectation taken in \(\MDP_{s_2}\)).
    \begin{alignat*}{3}
      & \Expectation[\langle \epsilon, s_1 \rangle](){\pfun\downarrow\phi_1}
      &
      \\
      & \leq \Expectation[\langle \epsilon, s_1 \rangle](){\lambda \epsilon' \cdot \Expectation[\langle \epsilon', s_2 \rangle]{\embed{\phi_2}}}
      & \Since{because \(\models \pbox[s_2]{\pfun\downarrow\phi_1} \, \phi_2\) and \cref{def:satisfaction}}
      \\
      & \leq \Expectation[\langle \epsilon, i \rangle]{\embed{\phi_2}}
      & \Since{\cref{lemma:expectation-properties}.\ref{item:nested-expectations}}
    \end{alignat*}
  \qed
\end{proof}

\begin{proof}[\Cref{thm:simple-statements}]
  \begin{enumerate}

    \item[\ref{thm:simple-statements}.\ref{item:skip}] Let us start from the left hand side, so assume that \( \epsilon \models \pbox[\code{skip}]{\func 1} \, \phi \). Then
      \begin{alignat*}{3}
        1
        & = \func{1}(\epsilon)
        &
        \\
        & \leq \Expectation[\langle \epsilon, \code{skip} \rangle]{\embed\phi}
        & \Since{\cref{def:satisfaction}}
        \\
        & = \inf_\pi \ExpectedV[\langle \epsilon, \code{skip} \rangle, \pi]{\embed\phi}
        & \Since{\cref{eq:expectation}}
        \\
        & = \ExpectedV[\langle \epsilon, \code{skip} \rangle]{\embed\phi}
        & \Since{\code{skip} deterministic, single policy}
        \\
        & = \sum_{\overline\sigma \in \textrm{paths}(\langle \epsilon, \code{skip}\rangle)} \hspace{-5mm} \Pr(\overline\sigma) \cdot (\embed\phi (\final{\overline\sigma}))
        & \Since{\cref{eq:expectation}}
        \\
        & = \sum_{\overline\sigma \in \textrm{paths}(\langle \epsilon, \code{skip}\rangle)} \hspace{-5mm} \Pr(\overline\sigma) \cdot (\embed\phi (\epsilon))
        & \Since[12mm]{definition of final, for \PDL \(\embed{\phi}(\langle\epsilon, \code{skip}\rangle) = \embed\phi(\epsilon)\)}
        \\
        & = 1 \cdot \embed\phi (\epsilon)
        & \Since{only one final path, singleton size}
      \end{alignat*}

      We have shown that \(1 \leq \embed{\phi}(\epsilon)\), which by definition of characteristic functions means that \(\epsilon \models \phi\). For the opposite direction, take \(\epsilon\models\phi\) and use the above six equalities from the bottom to show that \(\func{1}(\epsilon) = \Expectation[\langle \epsilon, \code{skip} \rangle]{\embed\phi}\), which means \(\epsilon \models \pbox[\code{skip}]{\func{1}} \,\phi \).

      \medskip

    \item[\ref{thm:simple-statements}.\ref{item:skip-composed}] One could argue from \cref{thm:sequencing}, but this is cumbersome, as it requires weakening the left-hand side to validity, which is not needed for \code{skip}, a special deterministic case. Thus it is better to prove directly from definition. The key step is that the infimum over all policies \(\pi\) for \(s\) is the same for \(\code{skip}; s\), because skip does not change valuation and the set of policies, see rule \textsc{Composition1}; so in this case neither the reward valuation or the policy can be chosen differently. The argument works in both directions.

    \medskip

    \item[\ref{thm:simple-statements}.\ref{item:assignment}] To prove this case we observe that executing a single assignment statement (see \textsc{Assign}) does not change the probability of final paths in the MDP associated with the program, so it does not change the expectation of the formula as long as the program is executed in the same valuation that the assignment creates.  This argument works in both directions:
      \begin{alignat*}{3}
          & \epsilon \models \pbox[ \code{x:=} e; s ]\pfun \, \phi
          &
          \\
          \text{iff }
          & \epsilon [x \mapsto \epsilon(e)] \models \pbox[\code{skip}; s]\pfun \, \phi
          & \Since{\textsc{Assign} does not change probability of paths}
          \\
          \text{iff }
          & \epsilon [x \mapsto \epsilon(e)] \models \pbox\pfun \, \phi
          & \Since{\cref{thm:simple-statements}.\ref{item:skip-composed}}
      \end{alignat*}

      \medskip

  \end{enumerate}
  \qed

  \end{proof}

\begin{proof}[\Cref{thm:loops-conditionals}]

  \begin{enumerate}

    \item[\ref{thm:loops-conditionals}.\ref{item:if-true}]  Let's write \(i\) for the program \(i=\IfThenElse\). We want to show that \(\pfun(\epsilon) \leq \Expectation[\langle \epsilon, i \rangle]{\embed\phi}\) where the expectation is taken in the MDP \(\MDP_i\). To do this we will reduce the calculation of the expectation to the MDPs \(\MDP_{s_1}\) and \(\MDP_{s_2}\) with the prefix for resolving the condition:
      \begin{alignat*}{3}
        & \pfun(\epsilon)
        &
        \\
        & \leq \Expectation[\langle \epsilon, s_1 \rangle]{\embed\phi}
        & \Since[180mm]{in \(\MDP_{s_1}\), assumption and \cref{def:satisfaction} twice}
        \\
        & = \inf_\pi \ExpectedV[\langle \epsilon, s_1 \rangle]{\embed\phi}
        & \Since[150mm]{in \(\MDP_{s_1}\), \cref{eq:expectation}}
        \\
        & = \inf_\pi \sum_{\sigma\in\paths{\langle\epsilon,s_1\rangle}} \Pr(\sigma) \cdot {\embed\phi}(\final{\sigma})
        & \Since[150mm]{\cref{eq:expectation}, path prob.\ taken under \(\pi\)}
        \\
        & = \inf_\pi \left [ 1 \cdot \hspace{-3mm} \sum_{\sigma\in\paths{\langle\epsilon,s_1\rangle}} \hspace{-3mm} \Pr(\sigma) \cdot {\embed\phi}(\final{\sigma}) \right ]
        & \Since[50mm]{trivial, also below}
        \\
        & = \inf_\pi \left [  1 \cdot \hspace{-4mm} \sum_{\sigma\in\paths{\langle\epsilon, s_1 \rangle}} \Pr(\sigma) \cdot {\embed\phi}(\final{\sigma}) +
        0 \cdot \hspace{-4mm} \sum_{\sigma\in\paths{\langle\epsilon, s_2 \rangle}} \Pr(\sigma) \cdot {\embed\phi}(\final{\sigma})  \right ]
        \\
        & = \inf_\pi  \sum_{\sigma\in\paths{\langle\epsilon, i \rangle}} \Pr(\sigma) \cdot {\embed\phi}(\final{\sigma})
        & \Since[150mm]{\textsc{If1}, assumption that \(\epsilon\models e\), \cref{def:satisfaction}}
        \\
        & = \Expectation[\langle \epsilon, i \rangle]{\embed\phi}
        & \Since[32mm]{in \(\MDP_i\)}
      \end{alignat*}
      Notice the shift of the MDP in the penultimate line to \(\MDP_i\) in the sum index. A line above, the assumption that the atomic formula \(e\) holds allows us to extend  the expectation by the suitable choice of the if-branch.  Since the if condition evaluates to true, the rule \textsc{if1} creates an MDP prefix advancing with probability one to state \(\langle \epsilon, s_1\rangle\) and with probability zero to the other state (we assume the same evaluation semantics for \LIT in the logic and for expressions in \pGCL). Afterwards we just observe that this is the same as calculating expected values directly from the if head.
      \looseness = -1

      \medskip

    \item[\ref{thm:loops-conditionals}.\ref{item:if-false}] The proof is symmetric to the previous case, just with the other \textsc{If} rule.

      \medskip

    \item[\ref{thm:loops-conditionals}.\ref{item:while-unfold}] This simple rule follows directly from operational semantics rules \textsc{While-1} and \textsc{While-2}:
      \begin{equation}
        \pfun \leq \Expectation[\langle \epsilon, i_1 \rangle]{\embed{\phi}} = \Expectation[\langle \epsilon, i_2 \rangle]{\embed{\phi}}
      \end{equation}
    where  \(i_1 = \whileunfolded\) and  \(i_2 = \whilefolded\).
    The equality of expectations holds because the two MDPs are identical---represent the same tree of terminating expectations over the same states, with the same branching structure, and final states. (Recall we only reason about almost surely terminating programs.) The \textsc{If-1} rule on the first program has exactly the same effect (and successors) as \textsc{While-1} rule on the second program.  Similarly both \textsc{If-2} on the first and \textsc{While-2} on the second program reduce to \code{skip} with the same valuation.

  \end{enumerate}
  \qed
\end{proof}

\subsection{Proofs for \cref{sec:purely-probabilistic}}

\begin{proof}[\Cref{thm:disjunction-to-conjunction}]
  Let \(s\) be a pGCL program,  \(\epsilon\) stand for a valuation, \( \func p \in \State \rightarrow [0, 1] \) be an expectation function, \(\pi\) a policy resolving non-determinism in \(s\), and \( \phi_i \in \PDL \) properties.  The expected reward \(\pfun\) after \(s\) for \(\embed{\phi_1\lor\phi_2}\)  under the scheduler \(\pi\) is a lower-bound for the sum of the separate expectations of \(\phi_1\) and \(\phi_2\), formally:
  \begin{equation}
    \pfun = \ExpectedV[\langle \epsilon, s \rangle, \pi]{\embed{\phi_1 \lor \phi_2}} \leq
    \ExpectedV[\langle \epsilon, s \rangle, \pi]{\embed{\phi_1}} + \ExpectedV[\langle \epsilon, s \rangle, \pi]{\embed{\phi_2}}\enspace.
  \end{equation}
  This is because in the rightmost sum, some of the states can be counted twice: once for \(\phi_1\) and once for \(\phi_2\).

  Note that, if $s$ is purely probabilistic (does not use the demonic choice operator) then the above can be written equivalently as follows: If \(\epsilon \models \pbox{\pfun} \, ( \phi_1 \lor \phi_2 ) \) then there exist \(\pfun_1\), \(\pfun_2\) such that \( \epsilon \models \pbox{\pfun_1}\, \phi_1 \) \textbf{and} \( \epsilon \models \pbox{\pfun_2}\, \phi_2 \) and \( \pfun_1 + \pfun_2 \geq \pfun \) everywhere.
  \qed

\end{proof}

\end{document}